\documentclass[conference]{IEEEtran}
\IEEEoverridecommandlockouts
\usepackage{cite}
\usepackage{amsmath,amssymb,amsfonts,amsthm}
\usepackage{algorithmic}
\usepackage{algorithm}
\usepackage{array}
\usepackage{graphicx}
\usepackage{textcomp}
\usepackage{xcolor}
\usepackage{amsthm}
\usepackage{subfig}
\usepackage{float}
\usepackage{enumerate}
\usepackage{tabularx}
\usepackage{stfloats}
\usepackage{url}
\usepackage{hyperref}
\usepackage{xcolor, soul, framed} 
\usepackage{xcolor}
\usepackage{verbatim}
\usepackage{pstricks}

\newtheorem{thm}{Theorem}[section]

\newtheorem{lem}[thm]{Lemma}
\newtheorem{prop}[thm]{Proposition}
\theoremstyle{definition}

\theoremstyle{remark}
\newtheorem{rem}[thm]{Remark}
\numberwithin{equation}{section}

\def\BibTeX{{\rm B\kern-.05em{\sc i\kern-.025em b}\kern-.08em
    T\kern-.1667em\lower.7ex\hbox{E}\kern-.125emX}}

\begin{document}
\renewcommand{\algorithmicrequire}{\textbf{Input:}} 
\renewcommand{\algorithmicensure}{\textbf{Output:}}
\title{Wireless Regional Imaging through Reconfigurable Intelligent Surfaces: Passive Mode\\
\thanks{National Foundation (NSFC), NO.12141107 supports this work.}
}

\author{\IEEEauthorblockN{Fuhai Wang\textsuperscript{1,2}, Chun Wang\textsuperscript{1}, Rujing Xiong\textsuperscript{1}, Zhengyu Wang\textsuperscript{1}, Tiebin Mi\textsuperscript{1}, Robert Caiming Qiu\textsuperscript{1}}
\IEEEauthorblockA{\textsuperscript{1} School of Electronic Information and Communications,\\ Huazhong University of Science and Technology, Wuhan 430074, China. \\
\textsuperscript{2} School of Artificial Intelligence and Automation,\\ Huazhong University of Science and Technology, Wuhan 430074, China. \\
Email:\{wangfuhai, wangchun, rujing, wangzhengyu, mitiebin, caiming\}@hust.edu.cn}
}

\maketitle
\begin{abstract}
In this paper, we propose a multi-RIS-aided wireless imaging framework in 3D facing the distributed placement of multi-sensor networks. The system creates a randomized reflection pattern by adjusting the RIS phase shift, enabling the receiver to capture signals within the designated space of interest (SoI). Firstly, a multi-RIS-aided linear imaging channel modeling is proposed. We introduce a theoretical framework of computational imaging to recover the signal strength distribution of the SOI. For the RIS-aided imaging system, the impact of multiple parameters on the performance of the imaging system is analyzed. The simulation results verify the correctness of the proposal. Furthermore, we propose an amplitude-only imaging algorithm for the RIS-aided imaging system to mitigate the problem of phase unpredictability. Finally, the performance verification of the imaging algorithm is carried out by proof of concept experiments under reasonable parameter settings. 
\end{abstract}


\section{Introduction} 
\IEEEPARstart{T}he technologies for integrated sensing and communication, energy efficiency, and intelligently aided communications have been identified as the main technology scenarios for 6G \cite{ITUR2022} in the International Telecommunication Union (ITU) white paper towards 2030 and beyond. On the one hand, these scenarios introduce new demands for perception in 6G. On the other hand, the subsequent key technologies of 6G also offer new methods for high-precision regional imaging. 


In communication sensing networks, array antennas are applied to estimate angles, by receiving signal phase difference. The application of existing large-scale array antenna technology\cite{wen2019survey} can provide a rich dimension of signal characteristics for single-station positioning. To achieve 3D spatial imaging, widely adopted geometric-based positioning methods leverage techniques like Angle of Arrival (AoA). This involves estimating terminal equipment position by measuring signal intersection points with a known direction. If candidate regions, determined by geometric data, converge at a singular point, precise regional imaging can be accomplished. 

As one of the potentially important technologies in the 6G standard\cite{ITUR2022}, reconfigurable intelligent surface (RIS) is an electromagnetic surface that can be regarded as a special type of relay "antenna" array, which dynamically adjusts parameters such as phase, amplitude, and frequency to influence and control the electromagnetic response of environmental objects. Due to its advantages of low power consumption, low cost, and high flexibility, it is more suitable to be deployed in large connectivity scenarios such as smart cities, smart factories to serve sensors, and so on. In communication\cite{basar2019wireless}, RIS is used to improve beam-forming gain, reduce interference, and minimize fading. In sensing\cite{song2023intelligent,chu2021intelligent,liu2022tdoa,meng2022intelligent}, RIS can act as a passive anchor point to provide geometric diversity, as well as to further improve the localization performance by controlling channel. RIS-aided imaging is more in line with the future trend of large-scale sensor data communication and the emergence of the Internet of Everything (IoE) system. Most of the researches\cite {rinchi2022single,zhang2021metalocalization,he2022high} focus on the design of phase-shift algorithms by using the property of reconfigurable phase-shift for regional localization imaging. 
\paragraph*{Contributions}
  First, we introduce the multi-block RIS mathematically concise linear channel modeling. Our focus lies on array signal processing to model the distributed multiple-input single-output (MISO) signals, as well as analyze the influence of each key parameter of the RISs on imaging performance. The characteristics and properties of the system imaging matrix are discussed. Specifically, we theoretically analyze the constraints on the number of samples, the total number of RIS units, and the position of multi-RIS on the property of the imaging matrix. The simulation results verify the correctness of the proposal. The second contribution of our study involves solving the problem that the phase data cannot be acquired to reconstruct the regional image. We propose an amplitude-only imaging algorithm for the multi-RIS-aided imaging system. To the best of our knowledge, we are the first to conduct conceptual experiments in real situations, and the results show the validity and robustness of the proposed algorithm.
\paragraph*{Notations}
Bold symbols in capital letters and small letters denote matrices and vectors, respectively. $\parallel \cdot \parallel $ denotes the Euclidean norm. The conjugate transpose, transpose and conjugate of~$\mathbf{A}$ are denoted by~$\mathbf{A}^H$ ,~$\mathbf{A}^T$ and ~$\mathbf{A}^{*}$, respectively. $\otimes$ is the Kronecker product. $\left| x \right|$ denotes the absolute value. $\nabla f(\mathbf{\cdot})$ means derivative of the function $f(\mathbf{\cdot})$. $\left\langle \mathbf{\cdot}, \mathbf{\cdot} \right\rangle$ means inner product of the vector.

\begin{figure}[htbp]
\centerline{\includegraphics[width=0.82\columnwidth]{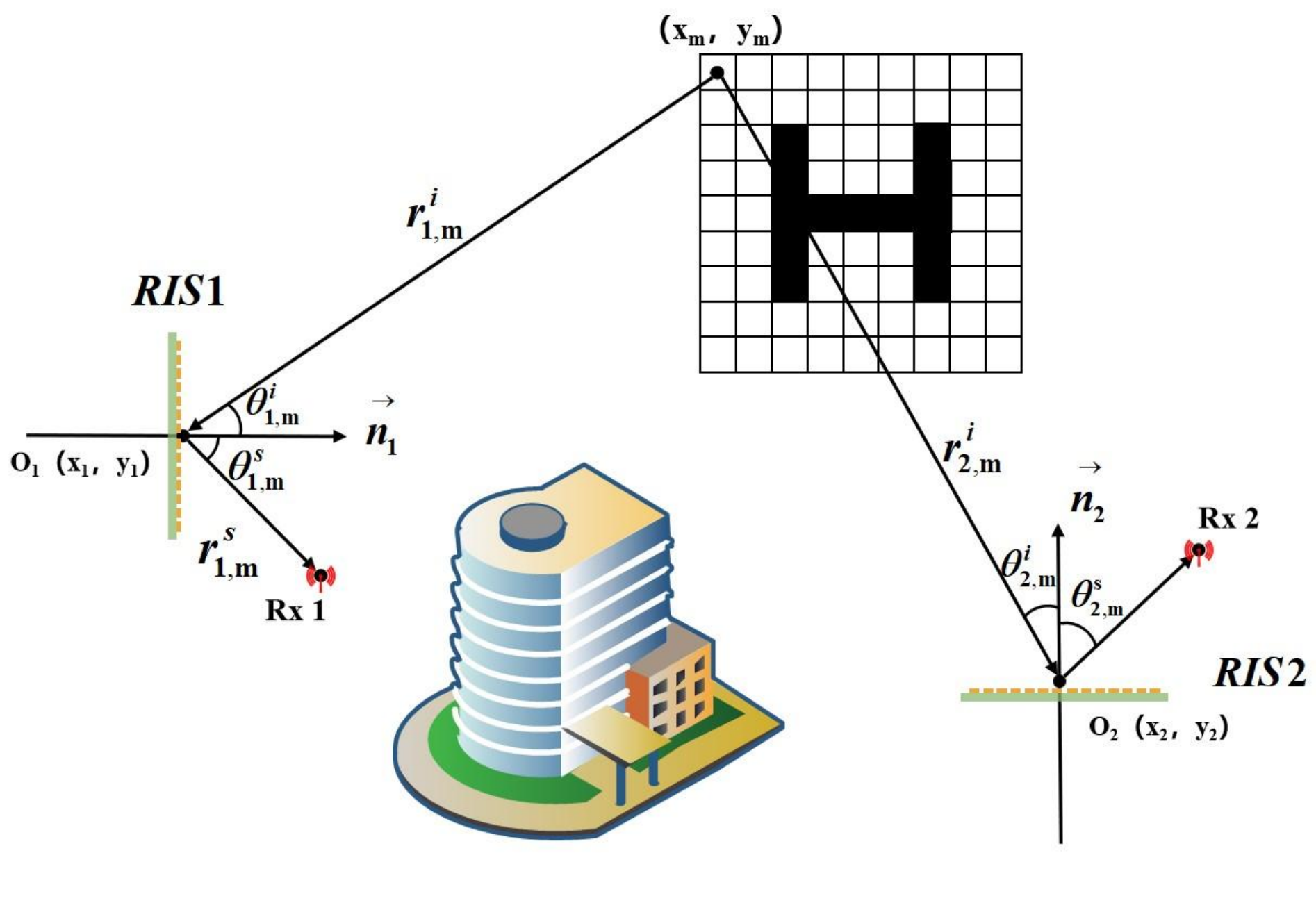}}
\caption{The RIS-aided passive imaging scene schematic.}
\label{1_fig}
\end{figure}

\section{System model and problem formulation}  
\subsection{Scenario Description}
The system model of the proposed RIS-aided imaging system is shown in Fig.~\ref{1_fig}. Considering imaging the 3D SoI in the far field, we have $\hat{d} \ge 2D^2/\lambda $ where $D$ is the maximum size of RIS\cite{cui2022near}. Each RIS is an N-elements uniform rectangular array (URA), and each element interval is $\lambda/2$, where $\lambda$ denotes the wavelength. Each RIS operates at a discrete phase shift of one bit serving a single receiver. There is an NLoS path between the SOI and the receiver, and a sensing link is provided through the RIS as a reflection. The 3D SoI is divided into M subspaces. Also, the complex-valued signal $\mathbf{s}\in \mathbb{C}^M $ from the SoI can be represented by a vector. The receiver and RIS operate passively at $f = 5.8 $ GHz.

\subsection{Channel Model of a Single RIS}
The M grid points of the SoI are considered incident signal sources. Considering that the RIS is a UPA, the channel for the single RIS-aided perception is modeled as follows\cite{mi2023towards}.
\begin{equation}
    \begin{aligned}
     \mathbf{h} = h^s  \mathbf{a}^s \mathbf{\Omega} \mathbf{A}^i \mathbf{H}^i . 
    \end{aligned}
\end{equation}
\noindent
Where $\mathbf{H}^i  = \text{diag}\left[\frac{ e^{-j 2 \pi r^i_1 / \lambda } }{ r^i_1 }, \cdots, \frac{ e^{-j 2 \pi r^i_M / \lambda } }{ r^i_M } \right]$ denotes the propagation channel matrix of M far-field sources to the RIS with phase and amplitude varying with distance. The distance from the RIS to the $m$-th source is denoted as $r^i_m$. The propagation channel between the RIS and the receiver is denoted as $ h^s = \frac{ e^{-j 2 \pi r^s / \lambda } }{ r^s }$. The distance from the RIS to the receiver is denoted as $r^s$. $\mathbf{A}^i$ denotes the incident steering vectors of the RIS unit array, which is expressed as follows:
$$
\mathbf{A}^i =
\begin{bmatrix}
  e^{ j 2 \pi \mathbf{p}_1^\top \mathbf{u_1}  / \lambda }  & \cdots  & e^{ j 2 \pi \mathbf{p}_1^\top \mathbf{u_M}  / \lambda } \\
  \vdots  & \ddots & \vdots \\
  e^{ j 2 \pi \mathbf{p}_N^\top \mathbf{u_1}  / \lambda }  & \cdots  & e^{ j 2 \pi \mathbf{p}_N^\top \mathbf{u_M}  / \lambda }\\
\end{bmatrix}.
$$
Then $\mathbf{a}^s = [ e^{ j 2 \pi \mathbf{p}_1^{\top} \mathbf{u}  / \lambda } \cdots e^{ j 2 \pi \mathbf{p}_N^{\top} \mathbf{u}  / \lambda ]^T }$ denotes the steering vector of reflections of RIS unit array. $\mathbf{p}_n^{\top} = [x_n, y_n, z_n]^T, n=1, \cdots,N$ is the spatial coordinates of each unit of RIS and $\mathbf{u}_m = [\sin \theta_m \cos \phi_m, \sin \theta_m \sin \phi_m, \cos \theta_m],m=1,\cdots,M$. In 3D space, the imaging system calculates the incident angle degrees $(\theta_m,\phi_m)$ from the coordinates of the $m$-th region position to the RIS. The RIS configuration is represented by a diagonal matrix $\mathbf{\Omega} = \text{diag} \bigl( \gamma_1 e^{j \Omega_1}, \cdots, \gamma_N e^{j \Omega_N} \bigr)$, where $\gamma_n$ and $\Omega_n$ represent the reflection amplitude and the phase shift coefficient of the $n$-th unit of RIS, respectively.  

\subsection{Received Signal}
There are M signals $\bigr(s_1e^{-j\alpha_{1}},s_2e^{-j\alpha_{2}},\cdots,s_Me^{-j\alpha_{M}} \bigr)$ from different grids $\bigr((\theta_1,\phi_1,r_1), \cdots,(\theta_M,\phi_M,r_M) \bigr)$ impinging on the RIS where $s_m$ denotes signal amplitude and $\alpha_{m}$ denotes signal phase. The received signal by the space-fed RISs can be formulated as
\begin{equation}\label{y_s}
    y^s = \mathbf{h} \mathbf{s}^i + n.
\end{equation}
We assume that the SoI is divided into M grids and $\mathbf{s}^i = [s_1e^{-j\alpha_{1}},s_2e^{-j\alpha_{2}},\cdots,s_Me^{-j\alpha_{M}}]^T$ represents the incident signals of electric filed. $n \in \mathbb{C}$ represents the Additive White Gaussian Noise (AWGN) that is distributed as $\mathcal{CN}(0,\rho)$.

Assume that the signal intensity distribution of the SoI during the short sensing time is approximately constant. By generating T code matrices {$\mathbf{\Omega}$} to sense the incident signals, the received signal matrix of T samples can be formulated as
\begin{equation}\label{y_kT}
    \begin{aligned}
         {\mathbf{y}}_{k,\mathbf{T}} 
          = {\mathbf{H}}_{k,\mathbf{T}} \mathbf{s}^i + \mathbf{n} 
          = \begin{bmatrix}
                \mathbf{h}_{k,1}\\
                \vdots \\
                \mathbf{h}_{k,T}\\
             \end{bmatrix} \mathbf{s}^i + \mathbf{n}. 
    \end{aligned}
\end{equation}
The interaction effect between the multiple RISs can be neglected. Hence, for a single receiver serving a single RIS distributed system, the signal received by the receivers of K systems with T sampling is given by
\begin{equation}\label{Y_{K,T}}
    \begin{aligned}
         \mathbf{y}^s_{\mathbf{K},\mathbf{T}} = \mathbf{H_{K,T}} \mathbf{s}^i 
        = \begin{bmatrix}
                {\mathbf{H}}_{1,\mathbf{T}}\\
                \vdots \\
                {\mathbf{H}}_{k,\mathbf{T}}\\
             \end{bmatrix} \mathbf{s}^i + \mathbf{n}.
    \end{aligned}
\end{equation}

With the parameters of the RIS-aided imaging system known, the imaging matrix $\mathbf{H_{K, T}}\in {\mathbb{C}^{K\times T \times M}}$ is known. The RIS-aided regional imaging can be transformed into the least squares (LS) solution problem.
\section{Theoretical analysis.}\label{SS:Properties of imaging matrix}
From the theory of computational imaging, as shown in Eq.(\ref{Y_{K,T}}), the rank of imaging matrix $\mathbf{H_{K,T}}$ influences the performance of a linear imaging system. Therefore, we would like to find the key parameters that constrain the RIS-aided imaging performance. 
\begin{thm}\label{T:1}
 The rank of imaging matrix $\mathbf{H_{K,T}}$ is not only affected by the number of unknown grid points M and the total number of samples K$\times$T, but it is also bounded by the total number of units K$\times$N.
 $$\text{rank}(\mathbf{H_{K,T}}) \le \min\{M,K\times T,K\times N\}$$
\end{thm}
\begin{proof}
    See Appendix~\ref{S:Appendix_A}.
\end{proof}
\begin{rem}
  It is necessary to set the number of controllable RIS elements K$\times$N and the total number of samples K$\times$T to be larger than M as much as possible so that the rank of imaging matrix $\mathbf{H_{K,T}}$ is not constrained by K$\times$N and K$\times$T to prevent rank deficit. We can view K$\times$T and K$\times$N as degrees of freedom sampled from both temporal and spatial dimensions. 
\end{rem}
Specifically, to investigate the effect of each parameter in the RIS deployment on the imaging matrix, the imaging matrix of one RIS can be referred to the representation in proof\ref{proof1} as follows:
$$\mathbf{H}_{k,\mathbf{T}} = \mathbf{W}_k \bf{P}_{\bf{k}}$$
\noindent
Each element $\omega_{k,n}^t$ in $\mathbf{W}_k$ denotes the reflection coefficient set at the $t$-th time for $n$-th element of $k$-th RIS. $\omega_{k,n}^t = \gamma_{k,n}^t e^{j \Omega_{k,n}^t} $ where $\gamma$ and $\Omega$ represent the reflection amplitude and the phase shift coefficient, respectively. 

Since the RIS is constructed using 1-bit, assuming that each unit has a reflection amplitude of 1 and a phase shift of $0$ or $\pi$, the matrix $\mathbf{W}_k$ can be regarded as an independently identically distributed matrix with each element obeying 1 or -1. The matrix $\mathbf{W}_k$ is a full rank matrix which is described in the following lemma \cite{ferber2021resilience}.

\begin{lem}\label{L:1}
Given an $n \times m$ matrix W with entries in $\{\pm 1\}$, if $n \leq m$, then W is of full rank with high probability(i.e. with probability $1-o(1)$). More precisely, if $m \ge n + n^{1- \epsilon /6}$, then even after changing the sign of $(1 - \epsilon )m/2$ entries, W is still of full rank with high probability.
\end{lem}

\begin{rem}
The above matrix in compressed sensing can be called the Rademacher sensing matrix. The matrix $\mathbf{W}_k$ belongs to the Rademacher sensing matrices. $\text{rank}\left( \mathbf{W}_k \right) = \text{min}\{T,N\} $ and the matrix $\mathbf{W}_k$ is full rank.
\end{rem}
The $ \bf{P}_{\bf{k}} $ matrix reflects the placement of RIS about the transmitter and receiver. It reflects how the deployment of RIS affects the performance of the imaging system. The matrix $ \bf{P}_{\bf{k}}$ is composed of four parts which, except for the $\mathbf{A}^i$ matrix, are diagonal matrices or constant. The diagonal elements are non-zero, so all can be considered full-rank or non-singular matrices. Only the properties of the $\mathbf{A}^i$ need to be considered.
\begin{prop}\label{T:2}
If two or more grid points in the SoI to the imaging RIS are on the same observation line, $\mathbf{A}^i$
will be deficient matrix, leading to the imaging matrix $\mathbf{H}_{k,\mathbf{T}}$ deficient. The performance of the imaging system is degraded.
\end{prop}
\noindent
When the grid points are at a similar angle of incidence concerning the RIS, the rank of this steering vector is defective, which leads to a defective imaging matrix. Distributed deployment of RISs can reduce the overlap of regions in the observation space.

In summary, to achieve high imaging performance, the RIS imaging system needs to consider that K$\times$T and K$\times$N are as large as possible to the unknowns in the SoI, M. The position of RIS determines the nature of matrix $\bf{P}_{\bf{k}}$, as well as the fact that deploying multiple RISs can ameliorate the shortcomings of a single RIS and a single-view deployment. In this paper, we are more concerned with the issue of system deployment. 

\section{Imaging system design and algorithm.}
\subsection{Fundamental performance analysis with the LS algorithm.}
We construct a set of linear equations using the statistical MISO behavior of the multiple RIS imaging system. Computational imaging can be accomplished using the LS algorithm to analyze the performance of imaging systems. The simulation modeling results are presented in the section~\ref{V}.

\subsection{Amplitude-only imaging algorithm}
In real systems, it is difficult to obtain the signal phase. Further, to address the constrained computational imaging in practical scenarios, we propose a specific phase retrieval algorithm.

The received signal matrix (\ref{Y_{K,T}}) is transformed into the following form with constraints:
\begin{equation}\label{abs(Y_{K,T})}
        \left|\mathbf{y}^s_{\mathbf{K},\mathbf{T}} \right|
        = \left| \begin{bmatrix}
                {\mathbf{H}}_{1,\mathbf{T}}\\
                \vdots \\
                {\mathbf{H}}_{k,\mathbf{T}}\\
             \end{bmatrix} \mathbf{s}^i + \mathbf{n}  \right|. \\
\end{equation}

We recover the signal strength distribution of the SoI using the received signal $\left| \mathbf{y}^s_{\mathbf{K},\mathbf{T}} \right|$, and the problem can be transformed into a phase retrieval problem. 
\begin{equation}\label{phase retrieval problem}
    \begin{aligned}
          \quad \text {find} & \quad  \mathbf{x}, \mathbf{x} \in \mathbb{C}^n \\
         \text { s.t. } & \left| \mathbf{Hx} \right| = \mathbf{b},  \mathbf{H} \in \mathbb{C}^{m \times n}, \mathbf{b} \in  \mathbb{R}^m  \\ 
    \end{aligned}
\end{equation}
Then the amplitude constraints are transformed into an intensity-based empirical loss function optimization problem as follows.
\begin{equation}\label{P1}
    \begin{aligned}
        (\mathrm{P} 1) \quad \min \quad  f(\mathbf{z}) := \frac{1}{2m}\sum_{i=1}^{m}(y_i - \left|  \mathbf{a}_i^H \mathbf{z} \right| )^2
    \end{aligned}
\end{equation}
where $\mathbf{z}$ is equal to $\mathbf{x}$ in equation \ref{phase retrieval problem}.
We try to solve the minimization problem of (P1) using iterative optimization, which has two parts.
\begin{algorithm}[H]
\caption{Amplitude-only Imaging Algorithm} \label{imaging algorithm1}
\begin{algorithmic}
\STATE 
\STATE {\textbf{Input:}}$\left\{\left\{y_{i}\right\}_{1 \leq i \leq m},\left\{\mathbf{h}_{i}\right\}_{1 \leq i \leq m},\left\{\eta_{i}\right\}_{1 \leq i \leq m}, T\right\}$
\STATE \hspace{0.5cm}$ \left\{\mathbf{h}_{i}\right\}_{i=1}^{m} $ : Sensing vectors  
\STATE \hspace{0.5cm}$ y_{i}=\left|\left\langle\mathbf{h}_{i}, \mathbf{x}\right\rangle\right|^{2} $ : measurements 
\STATE \hspace{0.5cm}$  \eta_{i} $ : the parameter 
\STATE \hspace{0.5cm}$ T $ : the maximum iteration times  
\STATE {\textbf{Output:}} $ \mathbf{x}^{*} $ 
\STATE \hspace{0.5cm}$ \mathbf{x}^{*}$  : the reconstructed signal  
\STATE {\textbf{Initialization}}
\STATE \hspace{0.5cm} set $ \rho^{2}=n \frac{\sum_{r} y_{i}}{\sum_{i}\left\|\mathbf{h}_{i}\right\|^{2}}  $
\STATE \hspace{0.5cm} set $ \mathbf{z}_{0},\left\|\mathbf{z}_{0}\right\|=\rho$  to be the eigenvector corresponding to the largest eigenvalue of 
\STATE \hspace{0.5cm} $$\mathbf{Y}=\frac{1}{m} \sum_{i=1}^{m} y_{i} \mathbf{h}_{i} \mathbf{h}_{i}^{*}  $$
\STATE \hspace{0.5cm} \textbf{for}  $k=1 $;$ k \leq  T $; $k++$  \textbf{do} 
\STATE \hspace{1.0cm} $\mathbf{z}_{k+1} \in \operatorname{argmin} f^{k}(\mathbf{z}) $
\STATE \hspace{0.5cm} \textbf{end for} 
\STATE \hspace{0.5cm} $ \mathbf{x}^{*}=\mathbf{z}^{T}  $
\end{algorithmic}
\end{algorithm} 

\begin{figure}[htbp]
\centerline{\includegraphics[width=0.55\columnwidth]{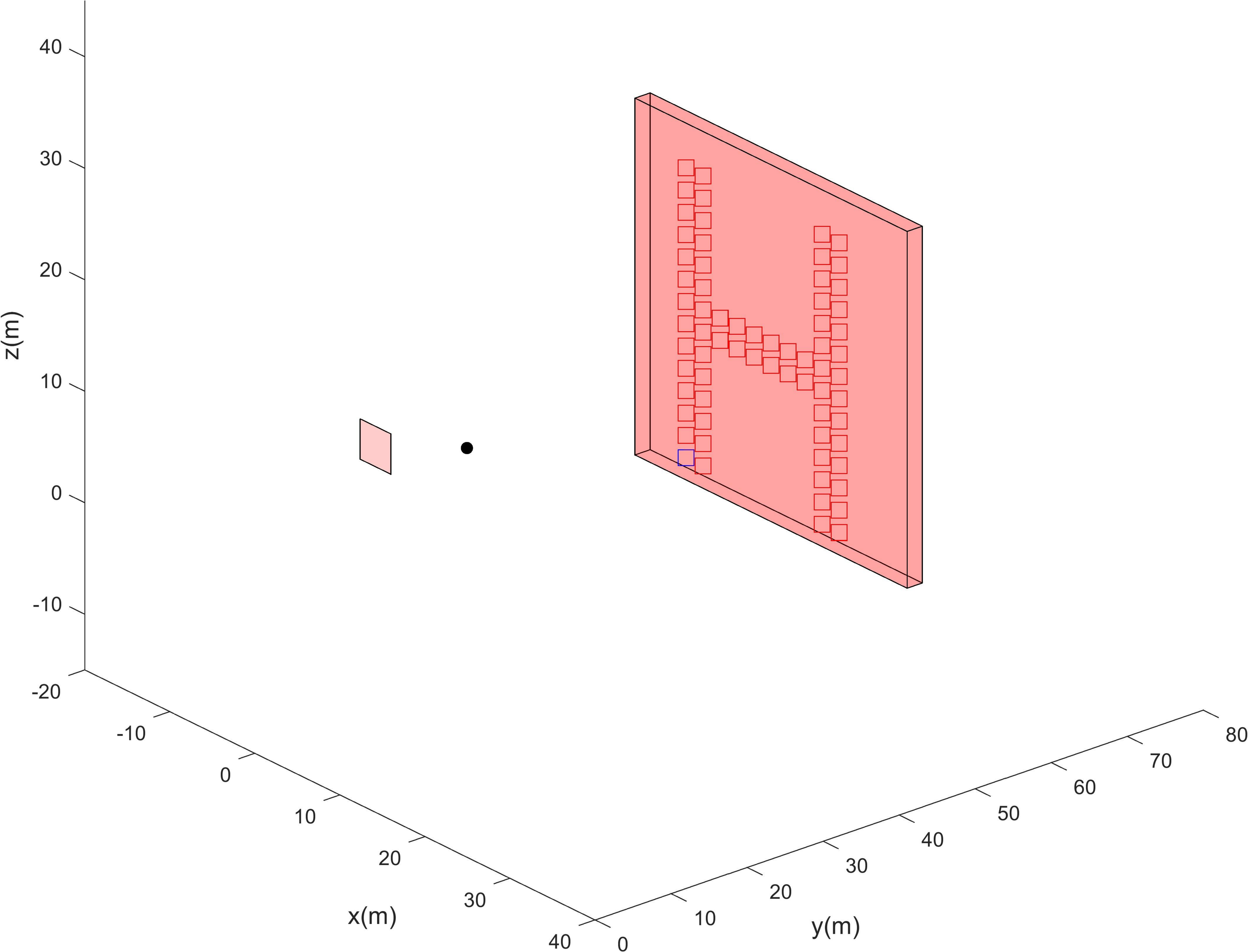}}
\caption{Schematic of the RIS-aided imaging in a 3D scene.}
\label{T_scene}
\end{figure}

Further, utilizing the iteratively reweighted LS algorithm, we can transform this solution into a reweighted wirtinger flow model to improve the computational performance.
\begin{equation}\label{P2}
    \begin{aligned}
    (\mathrm{P} 2) \underset{\mathbf{z} \in \mathbb{C}^{n}}{\min} f(\mathbf{z})=\frac{1}{2 m} \sum_{i=1}^{m} f_{i}(\mathbf{z})=\frac{1}{2 m} \sum_{i=1}^{m} \omega_{i}\left(\left|\left\langle\mathbf{h}_{i}, \mathbf{z}\right\rangle\right|^{2}-y_{i}\right)^{2},
    \end{aligned}   
\end{equation}
\noindent
where $\omega_i \ge 0 $ are weights. If $\{ \omega_i \}_{1 \le i \le m} $ are determined, \ref{P2} can be solved by gradient descent method.

Where $\omega_{i}^{k}=\frac{1}{\left.||\left\langle\mathbf{h}_{i}, \mathbf{z}_{k-1}\right\rangle\right|^{2}-y_{i} \mid+\eta_{i}}, i=1, \ldots, m.$ Those weights are adaptively calculated from the algorithm depending on the $\mathbf{z}_{k-1}$, where $\mathbf{z}_{k-1}$ is the result in the $(k-1)$-th iteration. $\eta_i$ is a parameter that can change during the iteration or to be stagnated all the time. The details can be seen in Algorithm1.

From Algorithm 1\ref{imaging algorithm1}, we can see that we will solve an optimization problem in each iteration $k$:
\begin{equation}
    \mathbf{z}_{k+1} \in \operatorname{argmin} f^{k}(\mathbf{z})=\frac{1}{m} \sum_{i=1}^{m} f_{i}^{k}(\mathbf{z})
\end{equation}
where $f_{i}^{k}(\mathbf{z})=\omega_{i}^{k}\left(\left|\left\langle\mathbf{h}_{i}, \mathbf{z}\right\rangle\right|^{2}-y_{i}\right)^{2} $.

For simplicity, we use a gradient descent algorithm to deal with it in this paper. The application of the amplitude-only imaging algorithm can be seen in Section~\ref{VI}.
\section{Performance analysis} \label{V}
We conduct simulations to analyze the impact of each key parameter. We discuss the impact of the total number of samples, the number of RIS units, and the position of multi-RIS on the imaging accuracy. Root-mean-square error (RMSE) and structural similarity index (SSIM) are considered metrics to evaluate the quality of the reconstruction image. 
\begin{figure}
\centering
\subfloat[N=15$\times$15.]{
\includegraphics[width=.32\columnwidth]{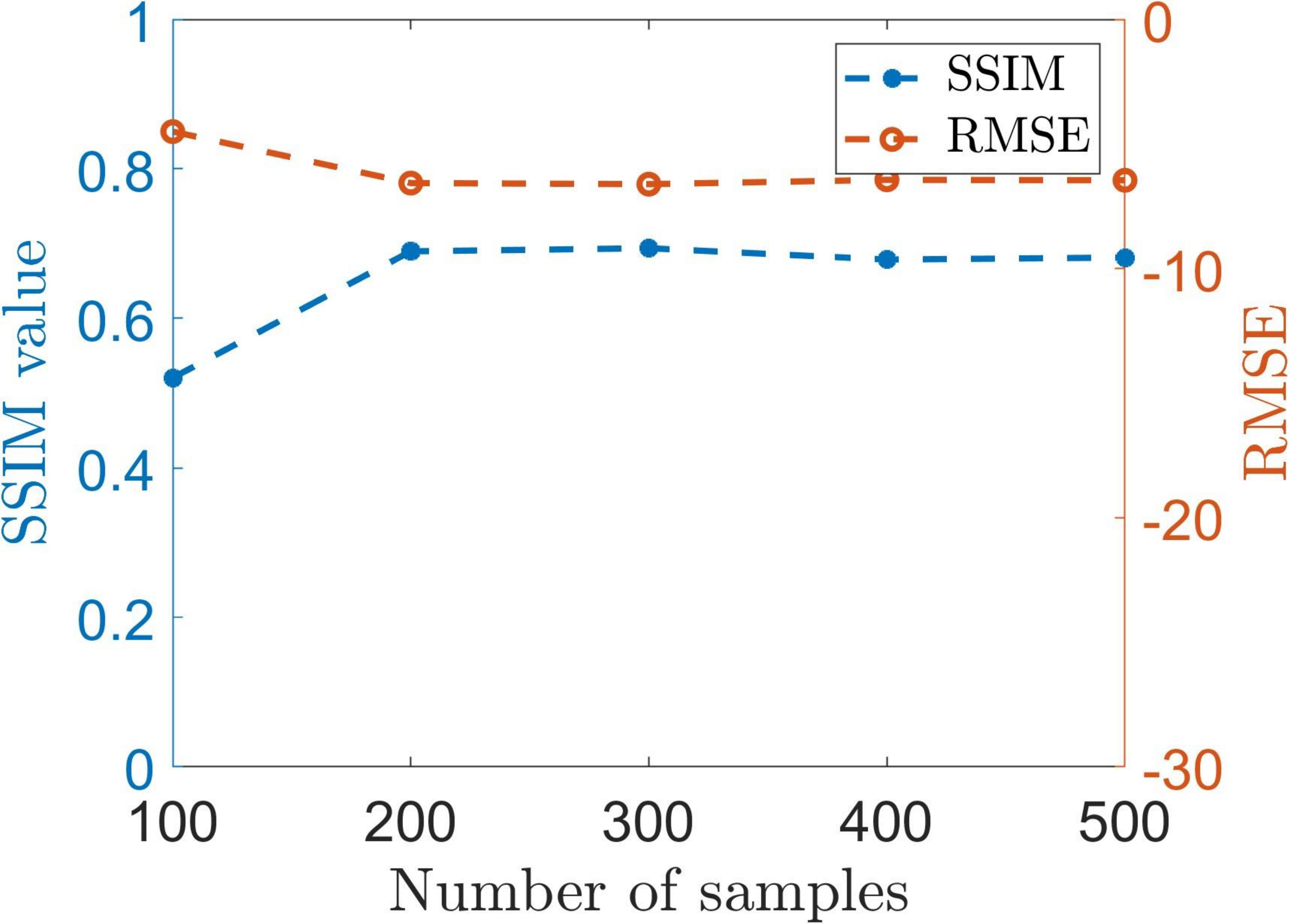}}
\subfloat[N=40$\times$40.]{
\includegraphics[width=.32\columnwidth]{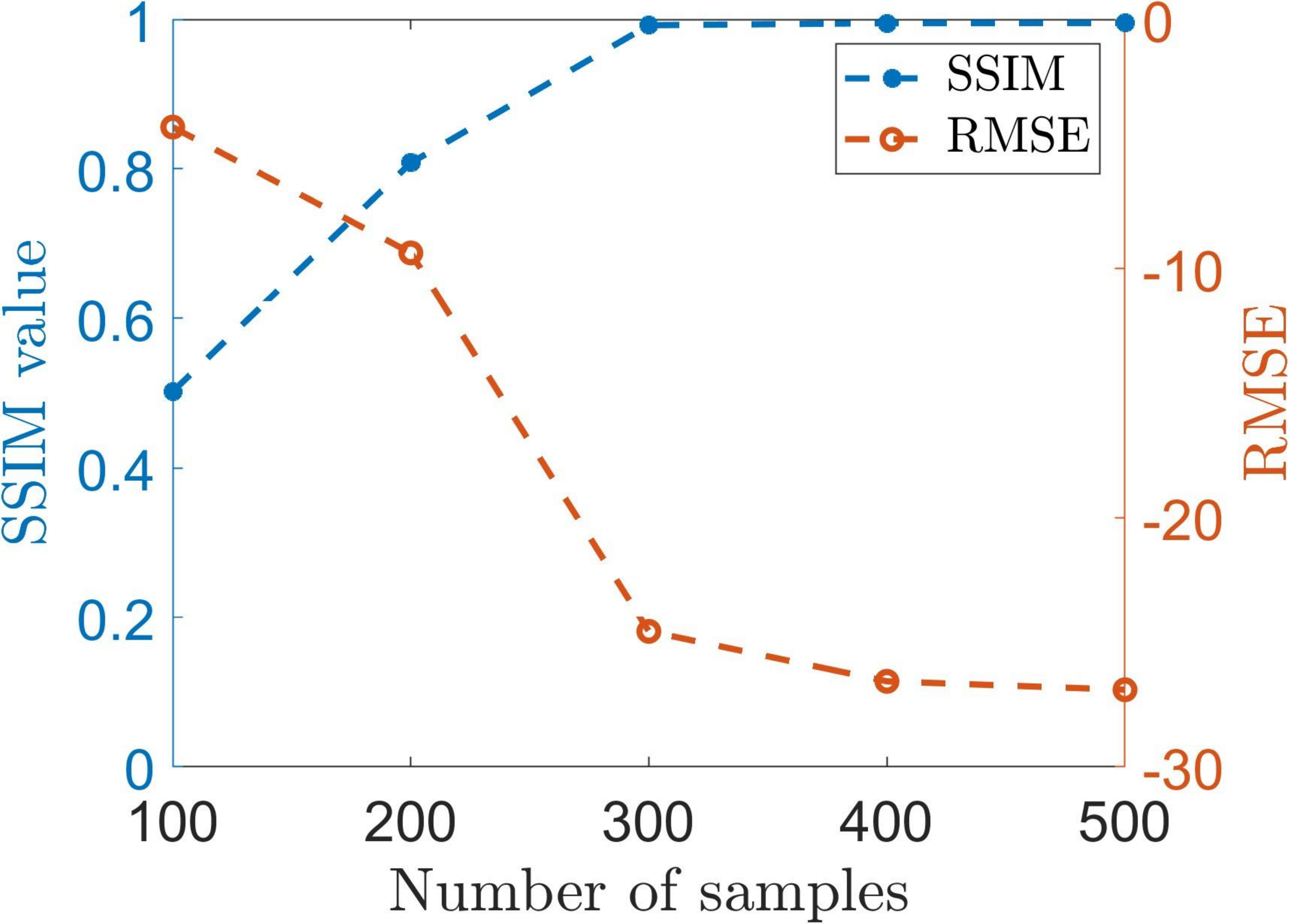}}
\subfloat[The rank of imaging matrixes.]{
\includegraphics[width=.3\columnwidth]{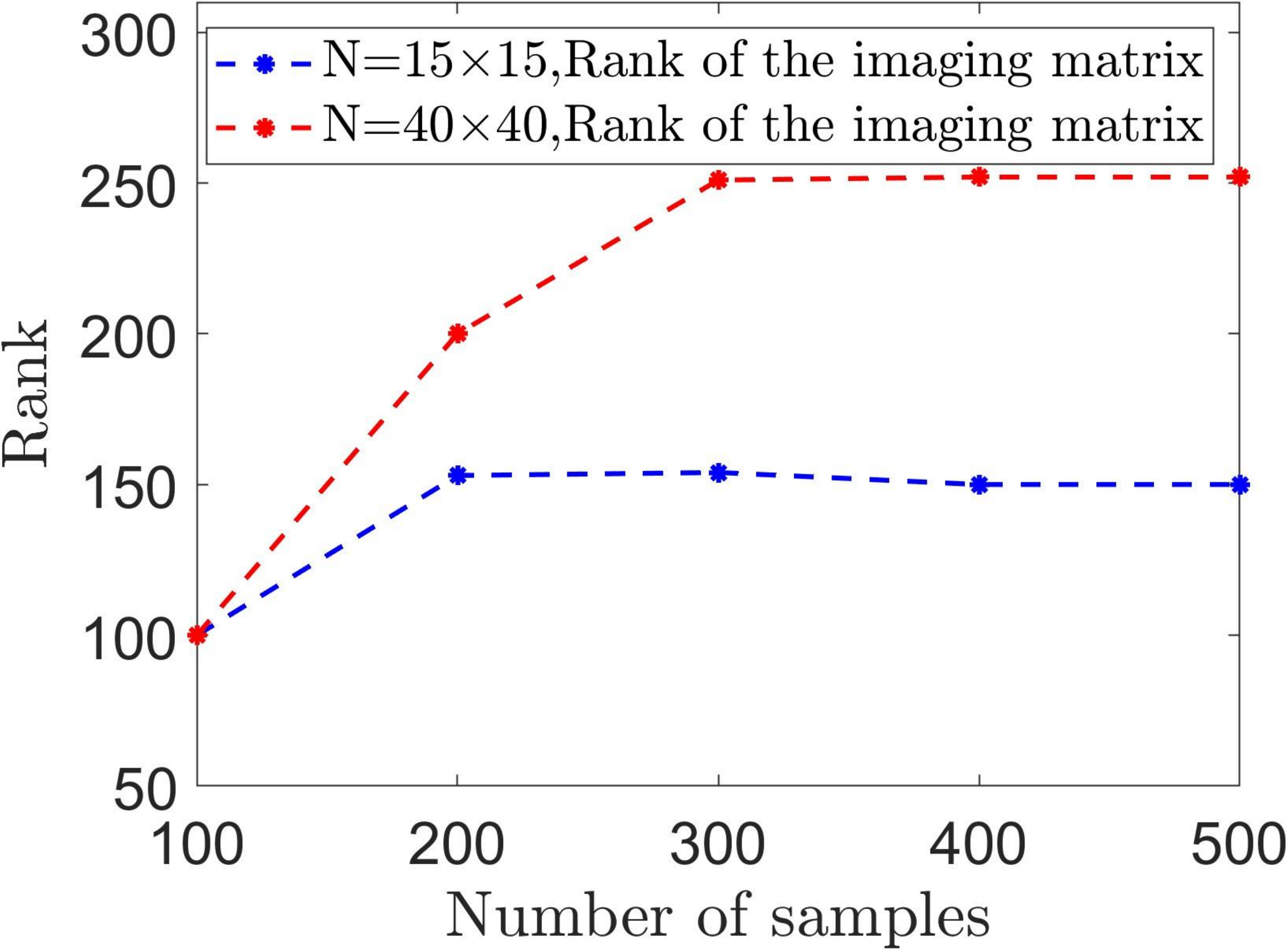}}
\caption{Performances with different numbers of samples.}
\label{T_SSIM}
\end{figure}

\begin{figure}[htbp]
\centerline{\includegraphics[width=0.8\columnwidth]{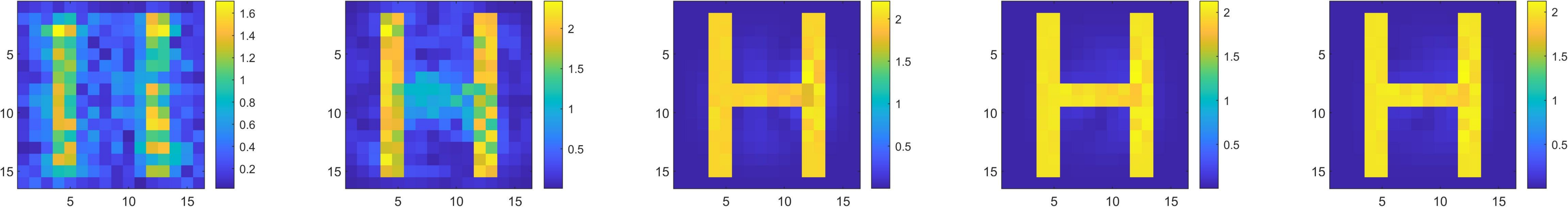}}
\caption{Results for different numbers of samples, N=40$\times$40.}
\label{T_N2}
\end{figure}

\subsection{Impact of the number of samples}\label{SS:Impact of the number of samples}
The number of sampled snapshots as a key parameter affects the accuracy of perception~\cite{9721164,zha2020doa}. Here, we consider two types of cases about the relationship between the number of RIS units N and the number of samples T. 

The scene is shown in Fig.~\ref{T_scene}, where a 32m $\times$ 32m $\times$ 2m SoI is discretized into M=16$\times$16$\times$1 at 2m intervals. We set the image of the letter "H" as the source, and the RIS is placed 50m in front of the SoI. The receiver is placed 10m in front of the RIS. We specify N as 15$\times$15 and 40$\times$40, respectively, and T is increased from 100 to 500 per interval of 100. The imaging results of N=40$\times$40 are shown in Fig.~\ref{T_N2}. The specific imaging results are shown in Fig.~\ref{T_SSIM}. When N is fixed, increasing T can improve the imaging quality. The SSIM gradually approaches 1 and RMSE gradually decreases in Fig.~\ref{T_SSIM}. When N=15$\times$15, is smaller, the imaging quality is not improving no matter how much T is increased. The result that the rank of the imaging matrix grows with T is shown in Fig~\ref{T_SSIM}(c). However, when T continues to be raised, the rank of the imaging matrix with N=15$\times$15 no longer improves.

As mentioned in Theorem\ref{T:1}, N bounds the rank of the imaging matrix. If N is smaller than M, the imaging matrix will be a rank deficit matrix, leading to a degradation of the imaging performance which we named spatial undersampling. When T is smaller than M, we can refer to this case as time undersampling. We further analyze the impact of the number of RIS units and RIS deployment on the imaging system.
\begin{figure}[htbp]
\centerline{\includegraphics[width=.6\columnwidth]{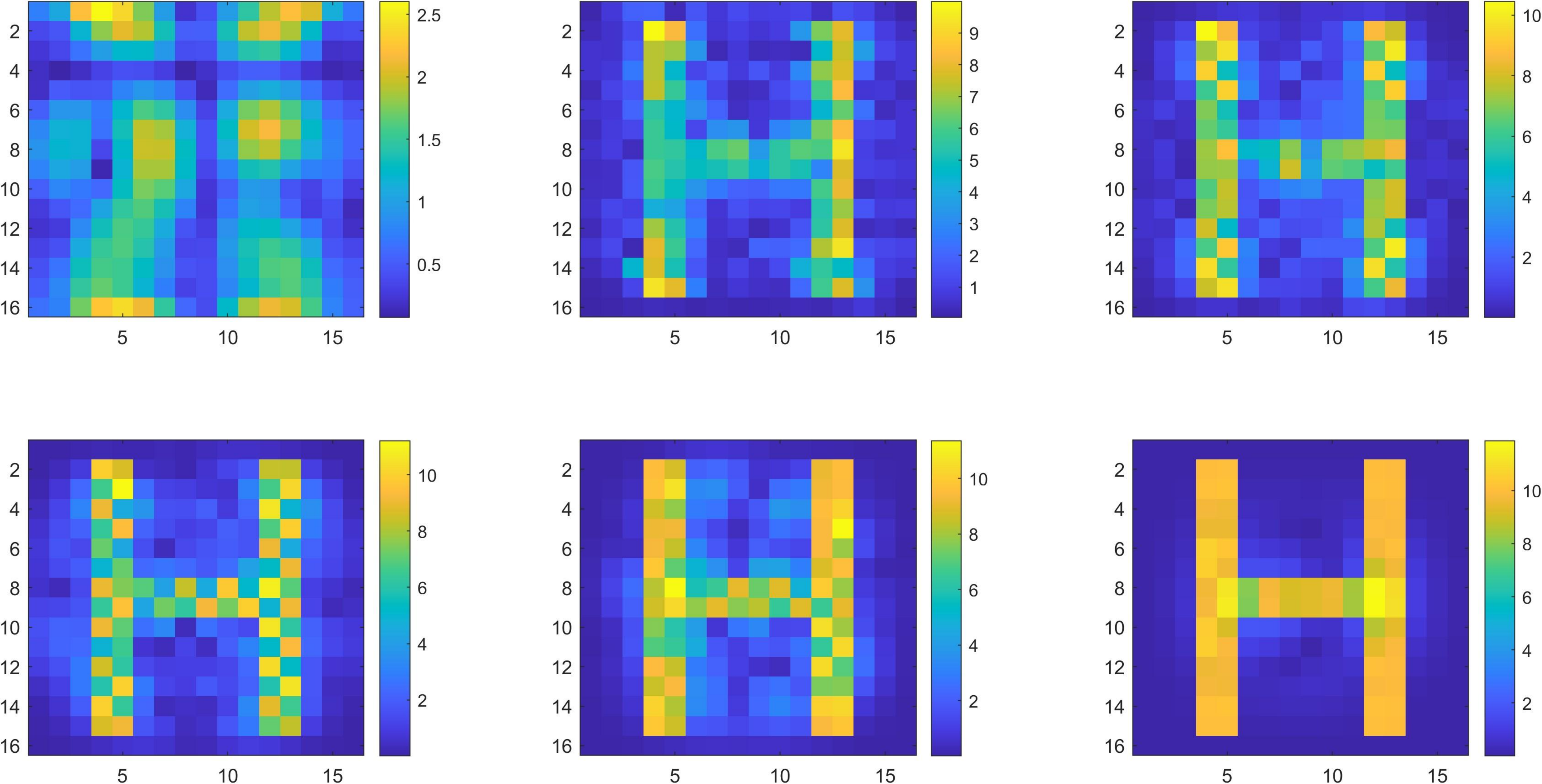}}
\caption{Imaging results for different numbers of RIS units}
\label{ls_matrix}
\end{figure}
\begin{figure}
\centering
\subfloat[Distribution of the singular values of imaging matrix.]{
\includegraphics[width=.3\columnwidth]{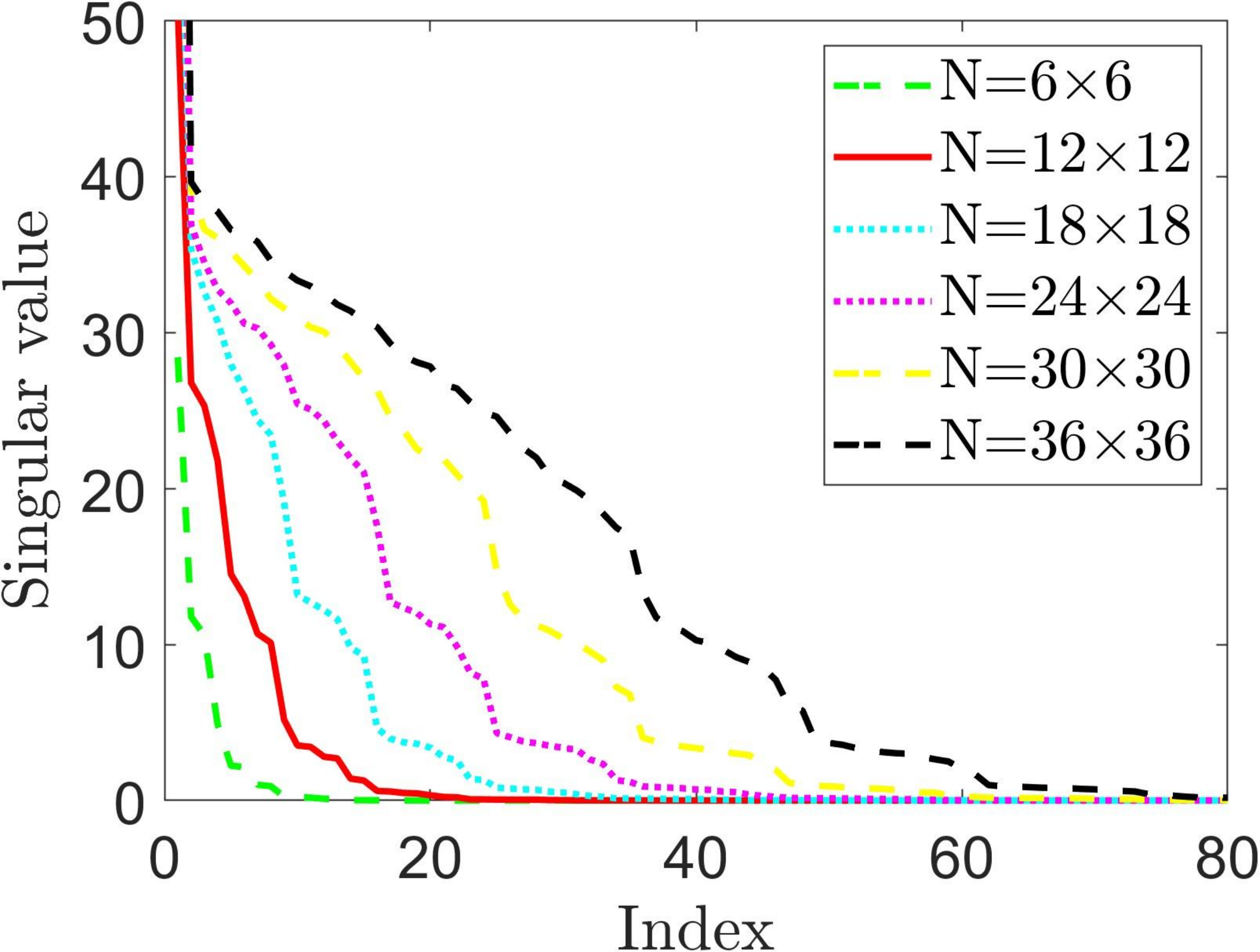}}
\subfloat[The rank of imaging matrix.]{
\includegraphics[width=.3\columnwidth]{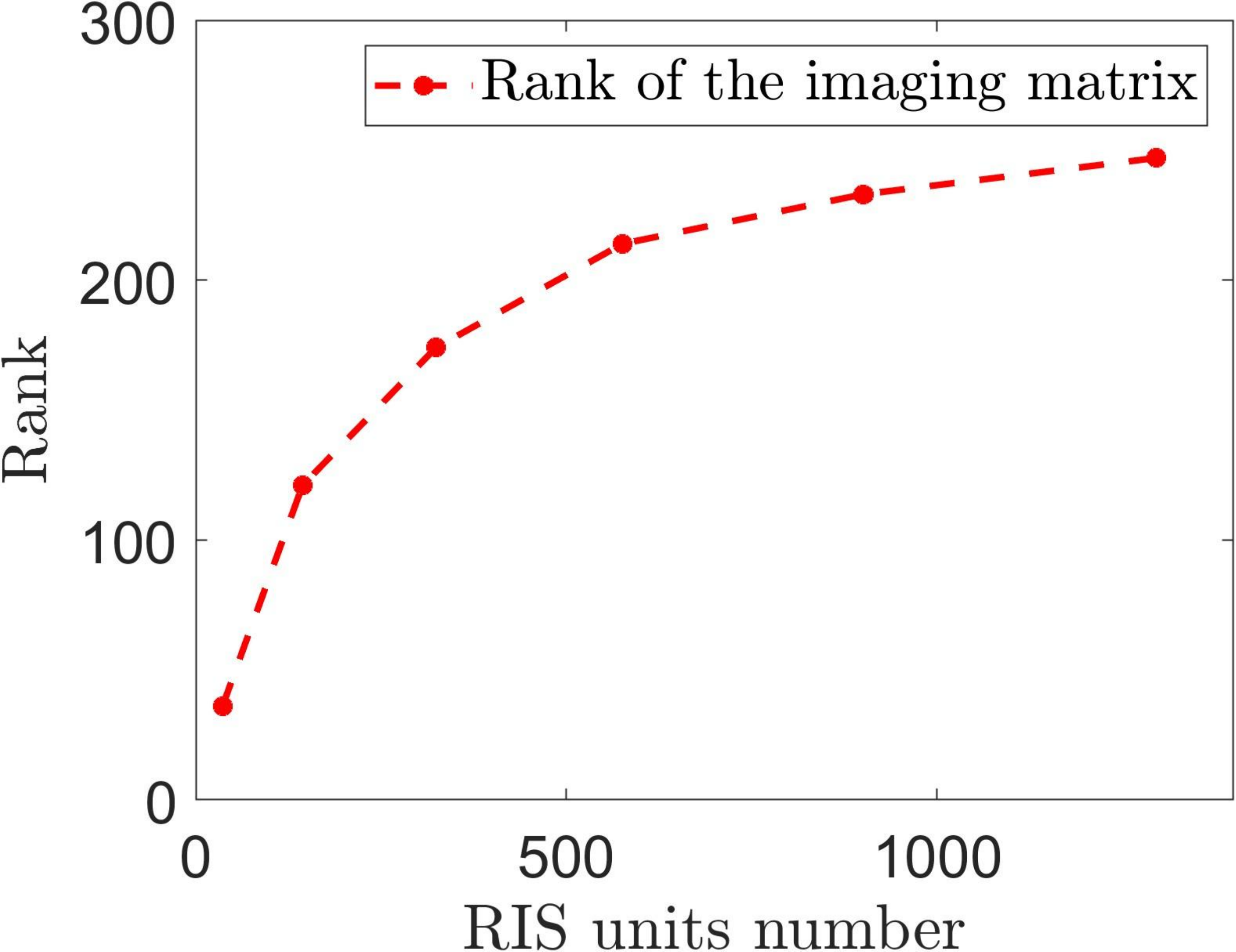}}
\subfloat[SSIM and RMSE with different numbers of the RIS units.]{
\includegraphics[width=.32\columnwidth]{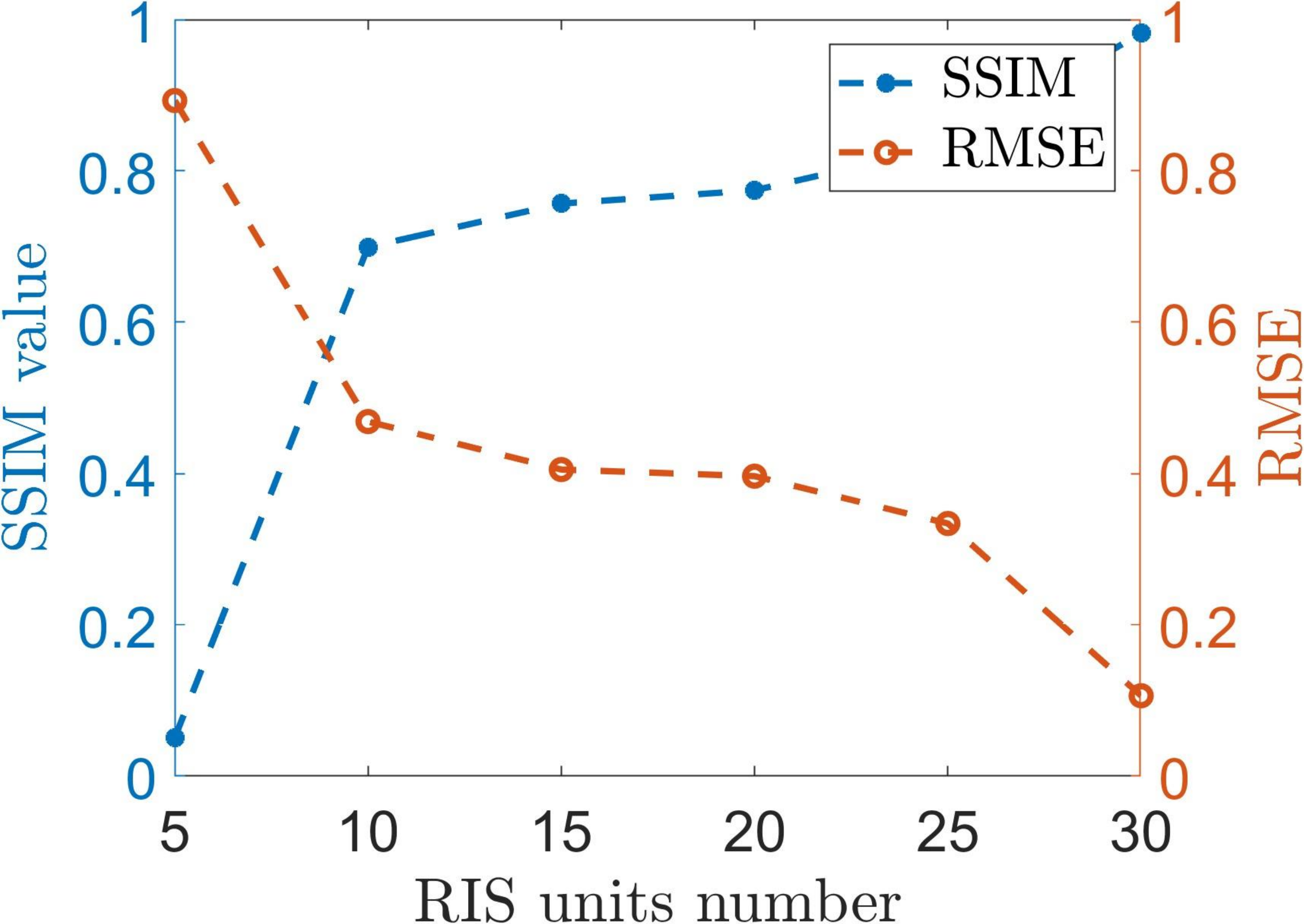}}
\caption{Performances with different numbers of RIS units.}
\label{N_svd_rank}
\end{figure}

\subsection{Impact of the number of RIS units} 

The number of RIS units N corresponds to spatial samplings. Assuming T$>$M, we try to change N and analyze the effect on the performance of the imaging system. Same as the previous scenario, T=300. We only change N from 6$\times$6,12$\times$12,18$\times$18,24$\times$24,30$\times$30 and 36 $\times$ 36, respectively. 

The simulations verify the results of the imaging effect as N varies in Fig.~\ref{ls_matrix}. The variation of singular value distribution of the imaging matrix with increasing N can be seen in Fig.~\ref{N_svd_rank}(a) and the rank of the imaging matrix also increases as is shown in Fig.~\ref{N_svd_rank}(b). It can be seen that with the increase of N, the singular value of the imaging matrix is overall larger, and the imaging image is more clear and accurate, which is conducive to the improvement of system imaging performance. Fig.~\ref{N_svd_rank}(c) shows the variation of SSIM and RMSE versus N.
\begin{figure}[htbp] 
\centerline{\includegraphics[width=.6\columnwidth]{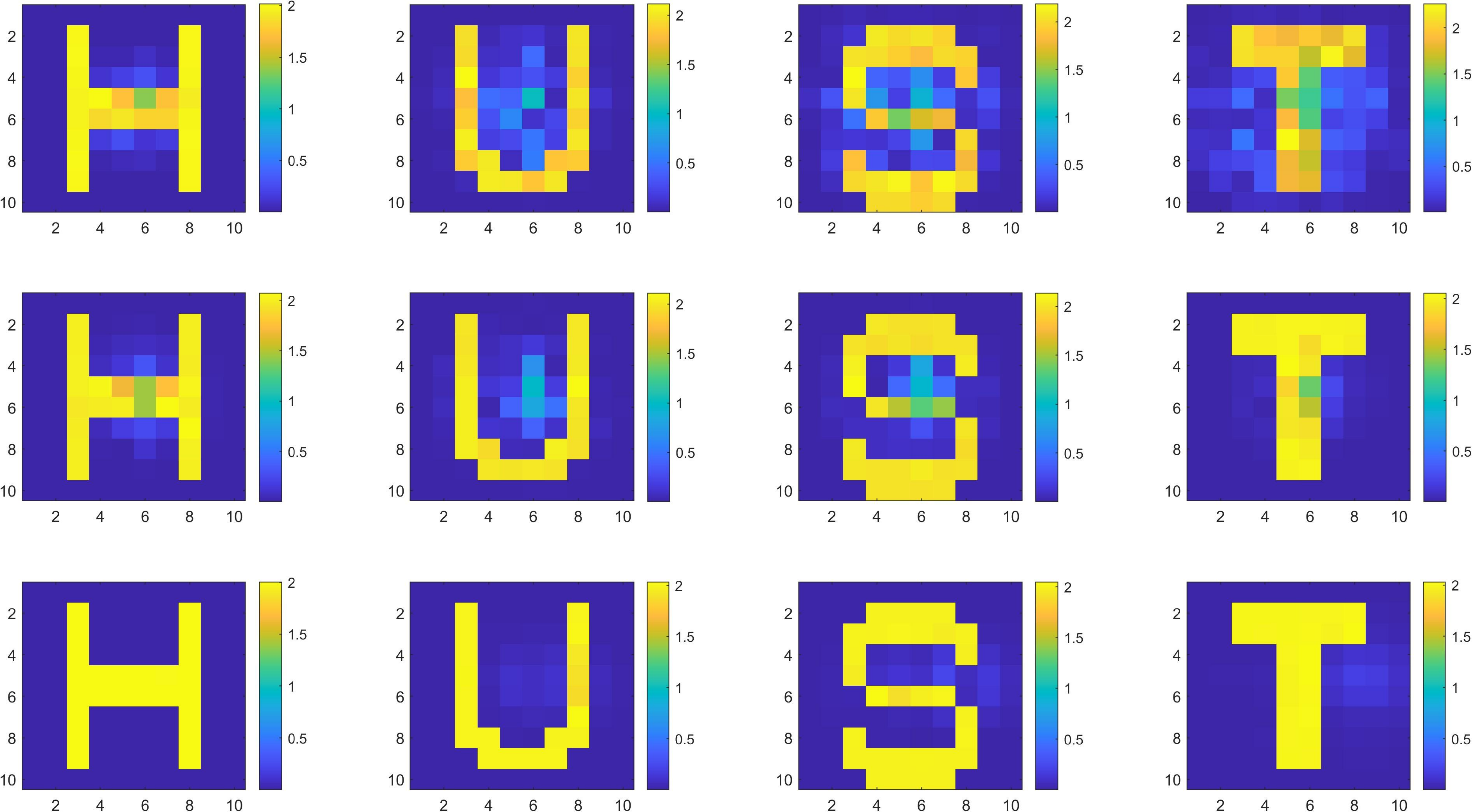}} 
\caption{Imaging results for different deployment scenarios.}
\label{p_ls_matrix}
\end{figure}

\begin{figure}
\centering
\subfloat[Distribution of singular values of the imaging matrix.]{
\includegraphics[width=.3\columnwidth]{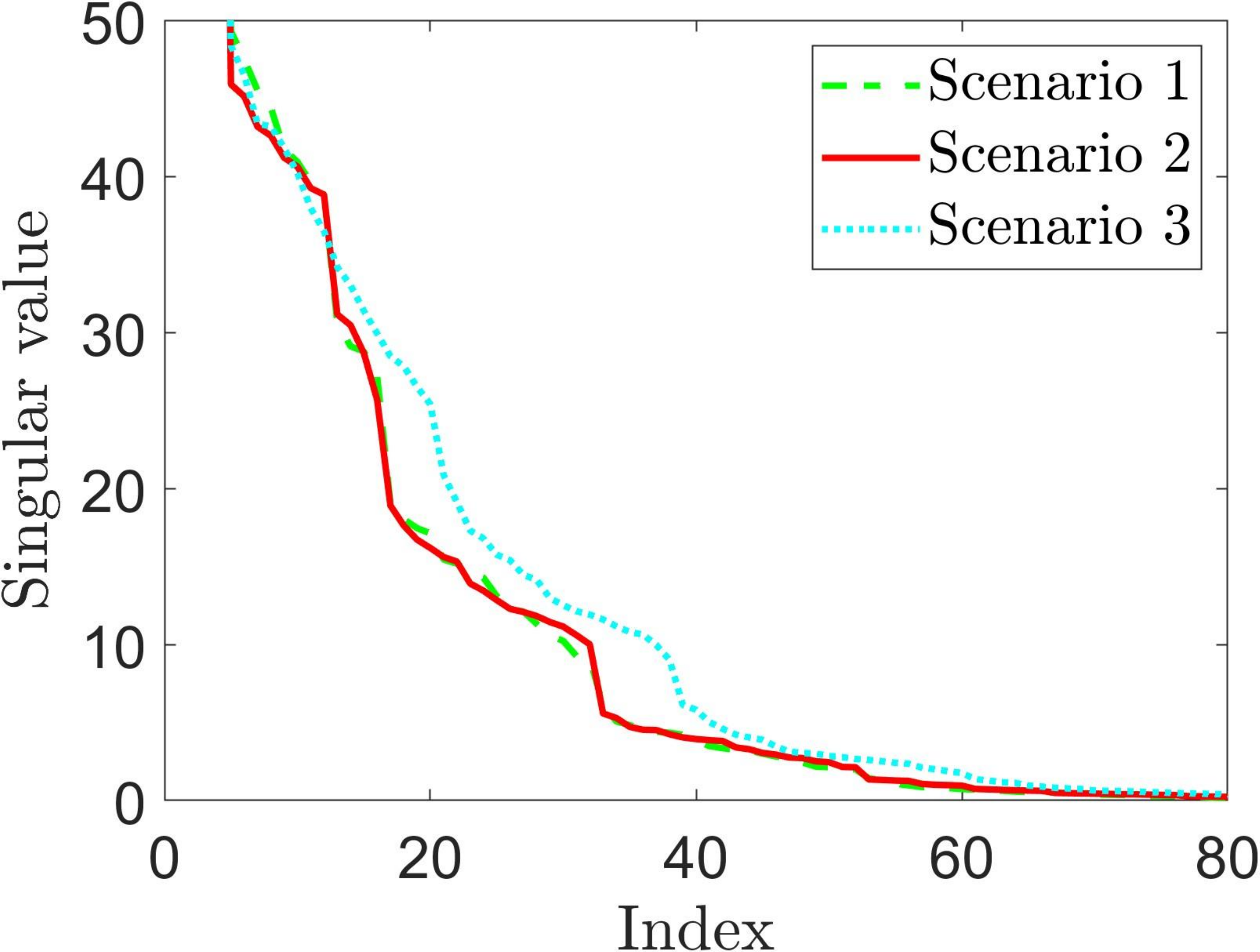}}
\subfloat[The rank of imaging matrix.]{
\includegraphics[width=.3\columnwidth]{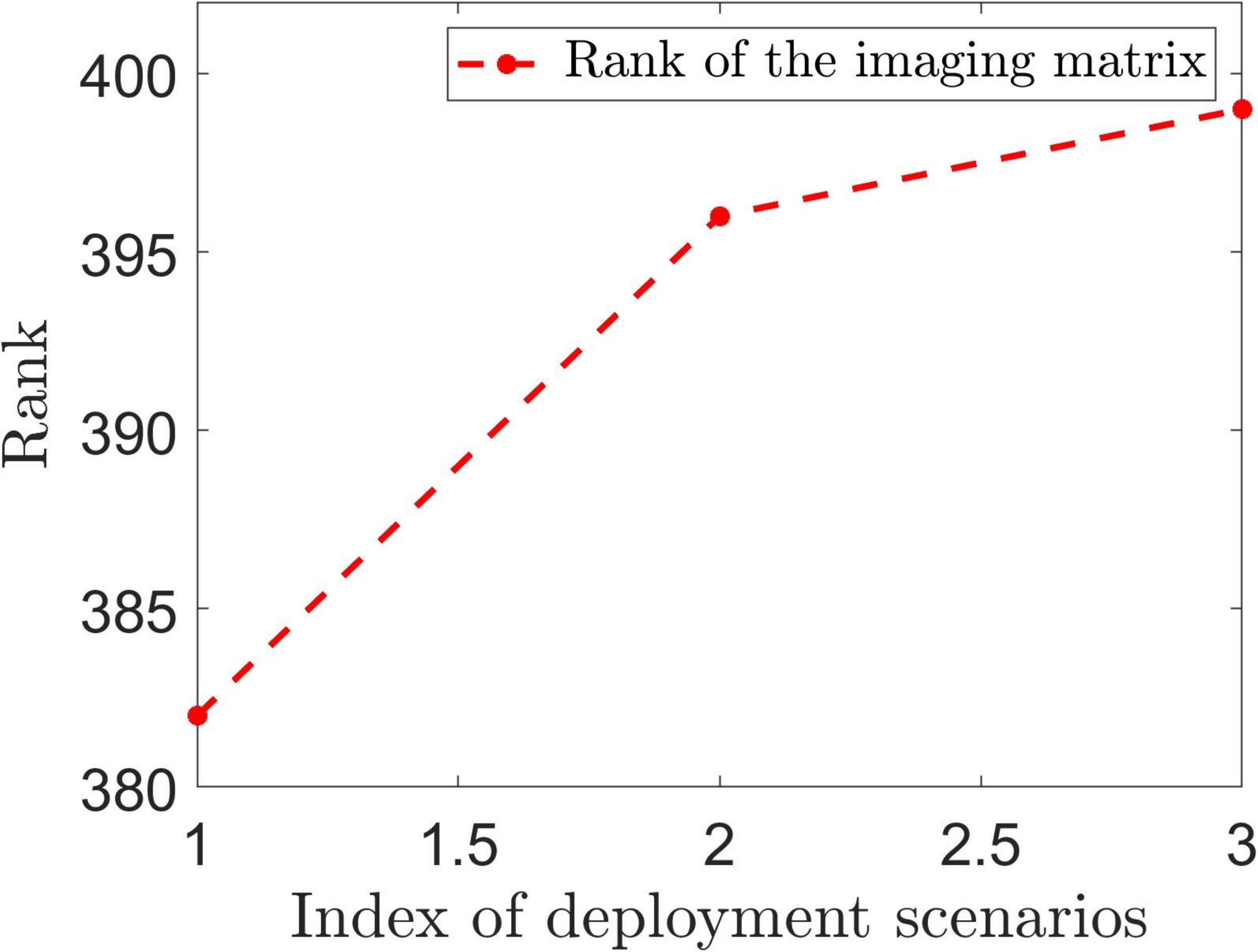}}
\subfloat[SSIM and RMSE with three deployment schemes.]{
\includegraphics[width=.33\columnwidth]{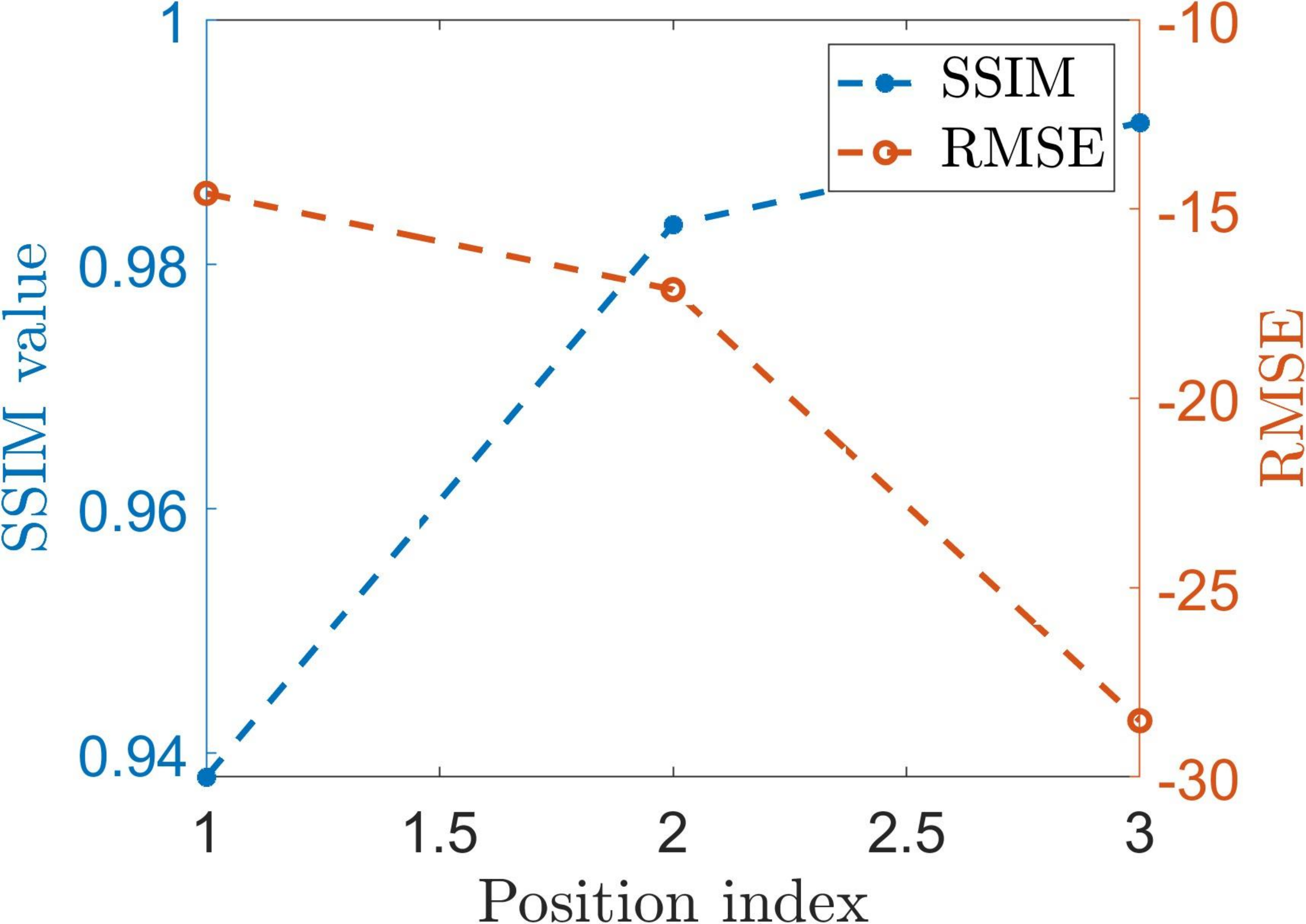}}
\caption{Performances of three scenes.}
\label{P_svd_rank}
\end{figure}

\subsection{Impact of multi-RIS locations}\label{SS:Impact of multi-RIS locations} 
For a single RIS, it isn't easy to distinguish the same incoming angle from the far field. Multi-RIS can improve the observation conditions of overlapping points. 
\begin{figure}
\centering
\subfloat[Scenario I.]{
\includegraphics[width=.32\linewidth]{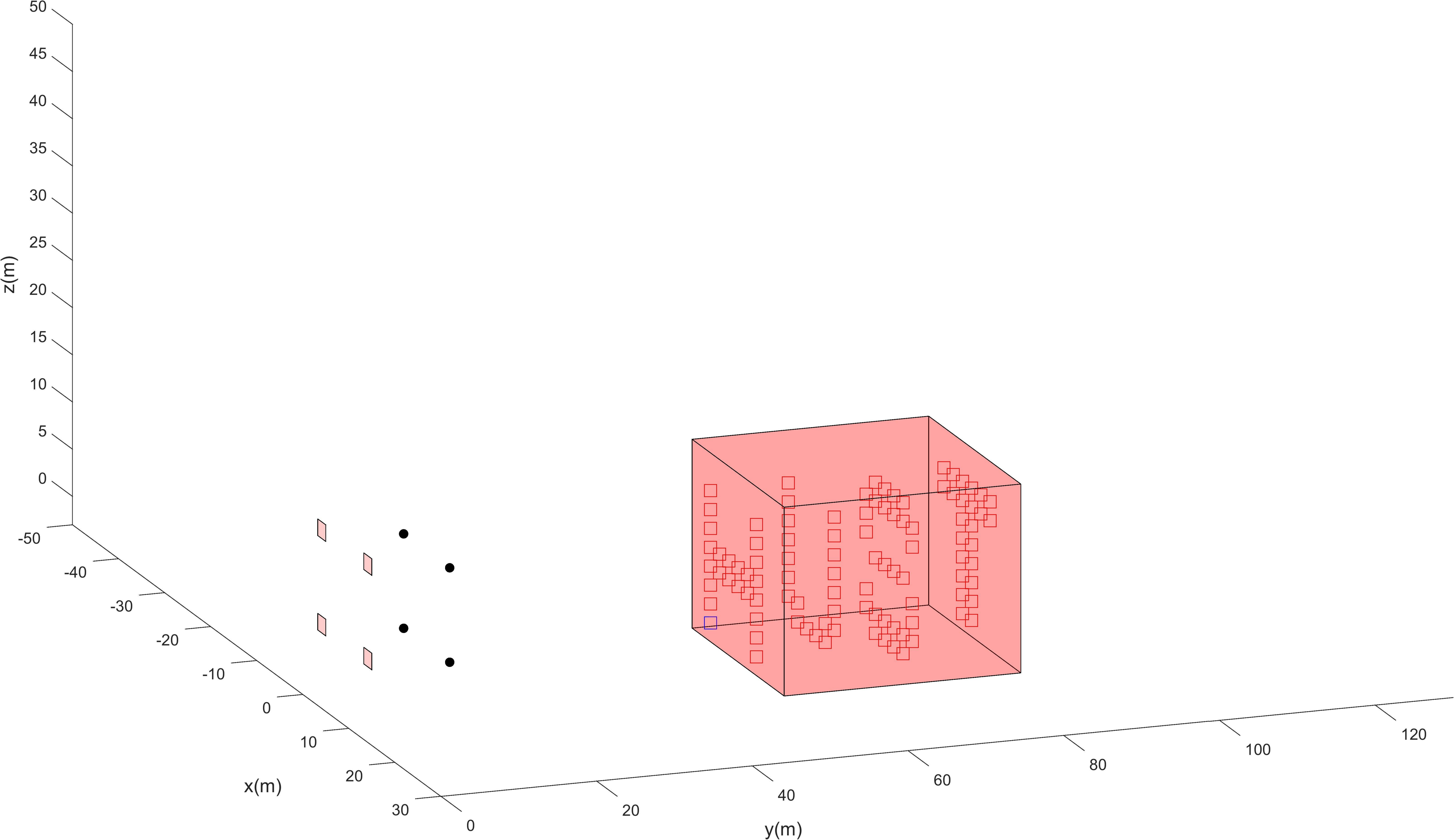}}
\subfloat[Scenario II.]{
\includegraphics[width=.32\linewidth]{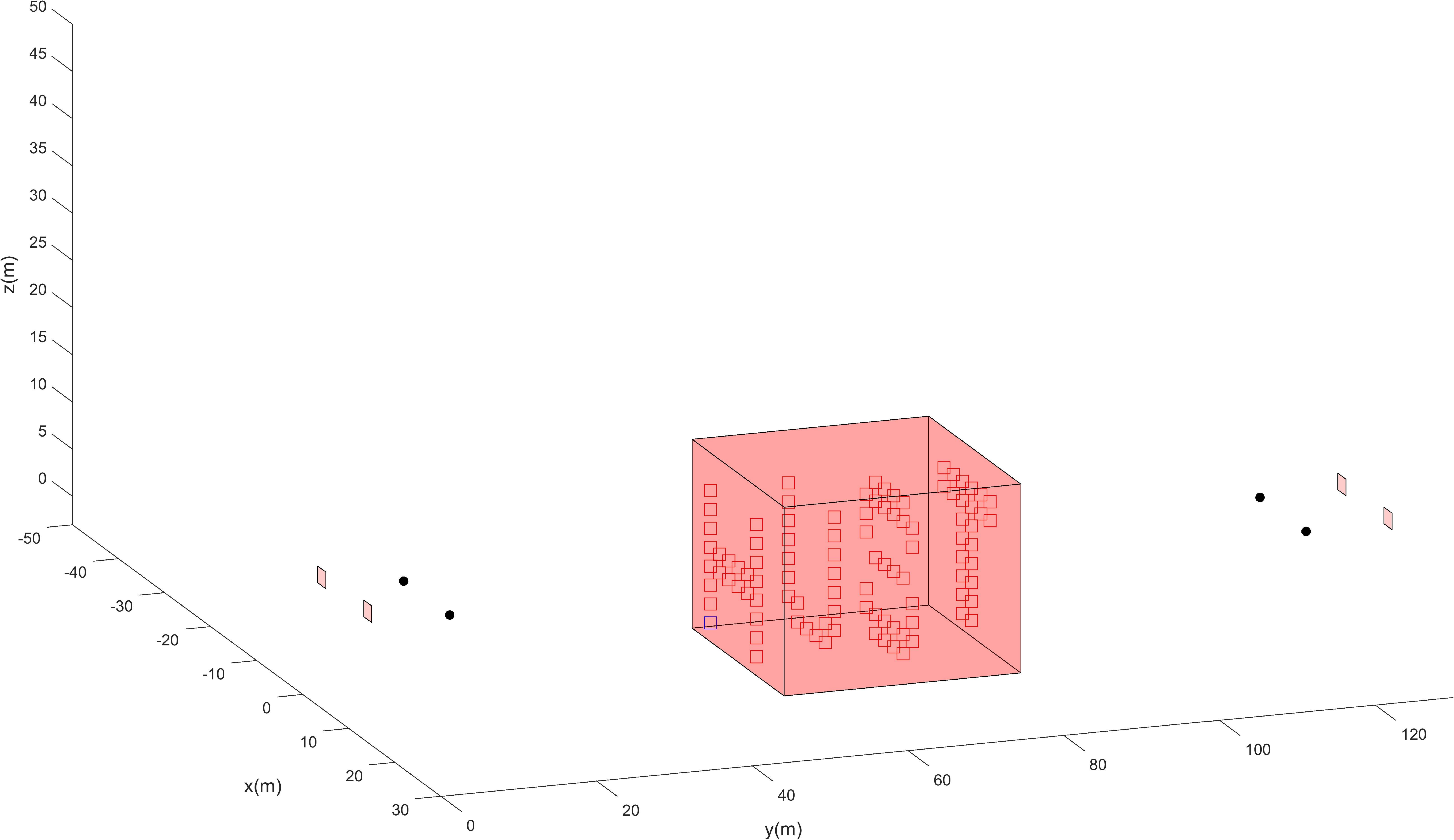}}
\subfloat[Scenario III.]{
\includegraphics[width=.32\linewidth]{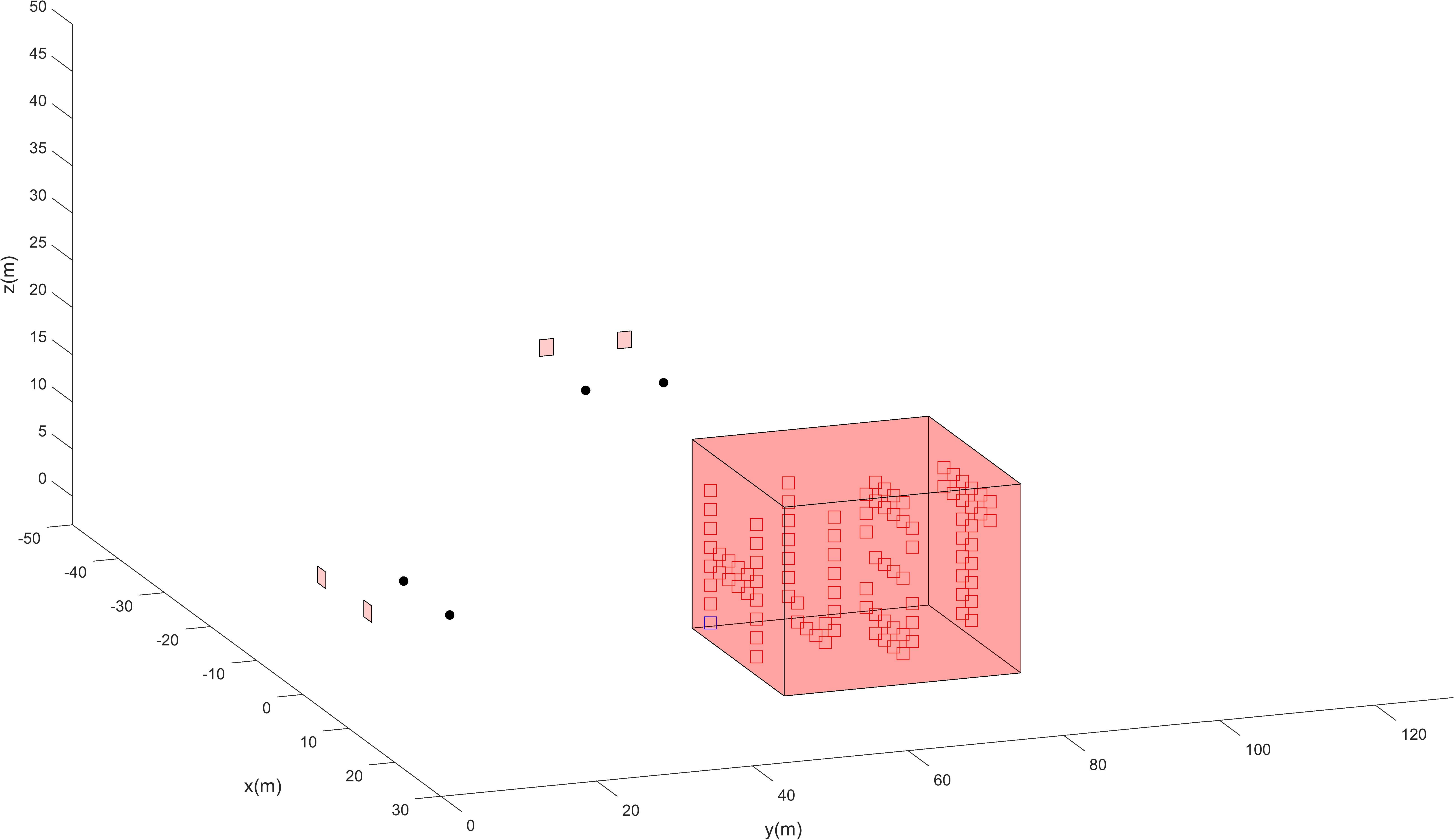}}
\caption{Different deployment scenarios for distributed RIS locations.}
\label{P}
\end{figure}
The SoI is divided into M=10$\times$10$\times$4 grids with dimensions 10m $\times$ 10m $\times$ 30m. The letters "HUST" are set as sources. We set up three schemes to place the four RISs, and N=20 $\times$ 20, which are shown in Figure~\ref{P}. 

The results of the three deployment schemes are shown in Fig~\ref{p_ls_matrix} and Fig~\ref{P_svd_rank}(c). The imaging quality can be improved when the RISs are placed vertically centered. In other words, the RISs are strategically positioned to maximize their coverage across various vertical observation surfaces. The singular values and the rank of the imaging matrix are improved, as shown in Fig~\ref{P_svd_rank}. 

When the incidence angles from different grid points to the RISs are the same, it results in too low a rank of the P-matrix and makes it difficult to recover the original signal. A reasonable deployment scheme for imaging from multiple angles, with as wide an angle of view of the SoI as possible, can solve the above problem. 

\section{Validation of experiment for amplitude-only.} \label{VI}



\subsection{proof-of-concept experiment for amplitude-only imaging}\label{Estimation Experiment}
We design two proof-of-concept experiments. Data from two real-world scenarios are collected by the Universal Software Radio Peripheral. The designed prototype of RIS can be found in\cite{xiong2023ris}. The RISs complete a single measurement of the SoI by changing the phase shift through the FPGA. Computational imaging is executed on the 2D SoI. For simplicity, we set the state of each column of RIS to be the same so that the RIS with 32$\times$10 elements can be regarded as a 32$\times$1 line array. The experiment scenes are shown in Fig.\ref{Experiment scene}.

\begin{figure}
\centering
\subfloat[Imaging in the chamber.]{
\includegraphics[width=.4\columnwidth]{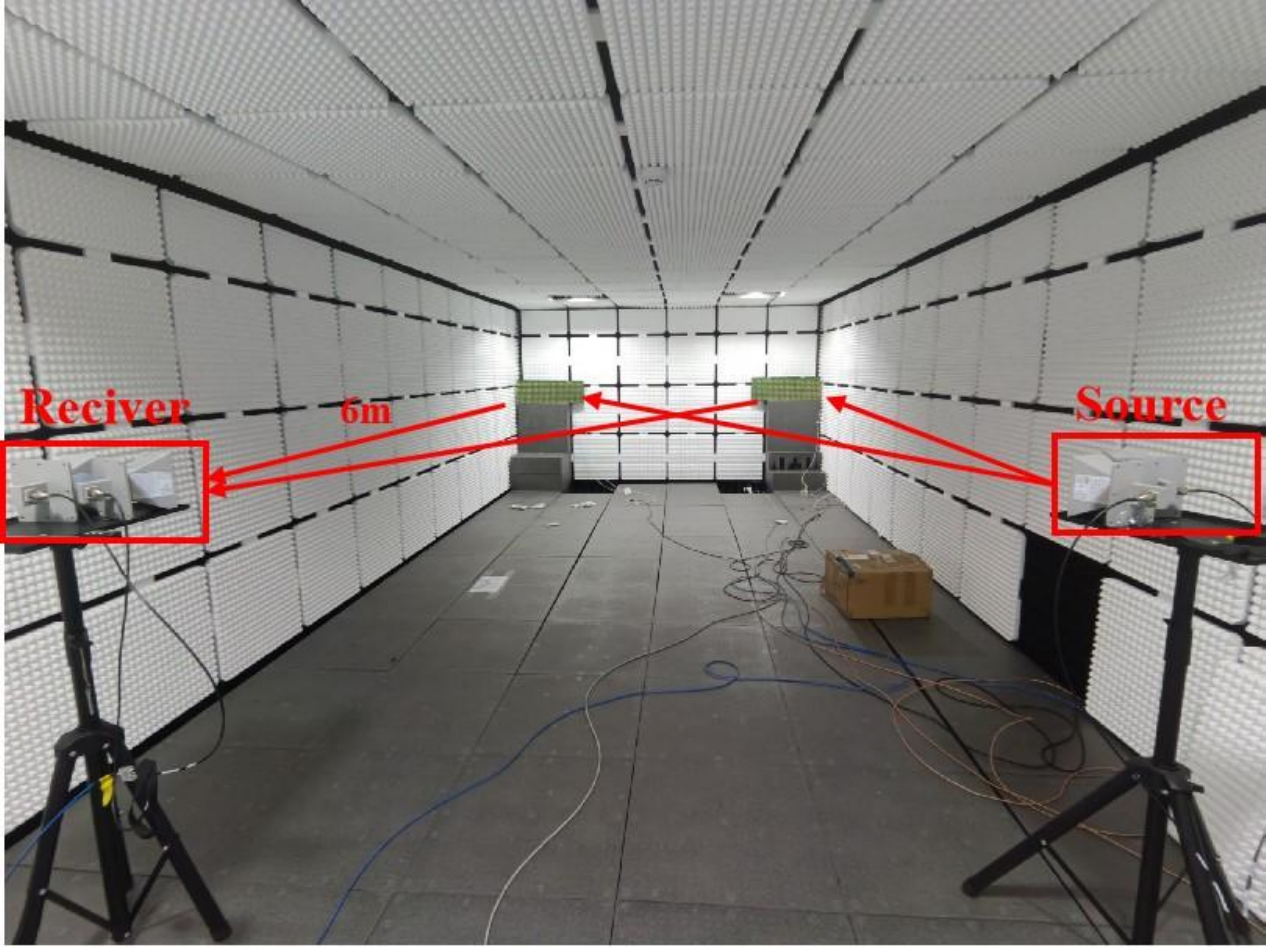}} 
\subfloat[Imaging in the office.]{
\includegraphics[width=.4\columnwidth]{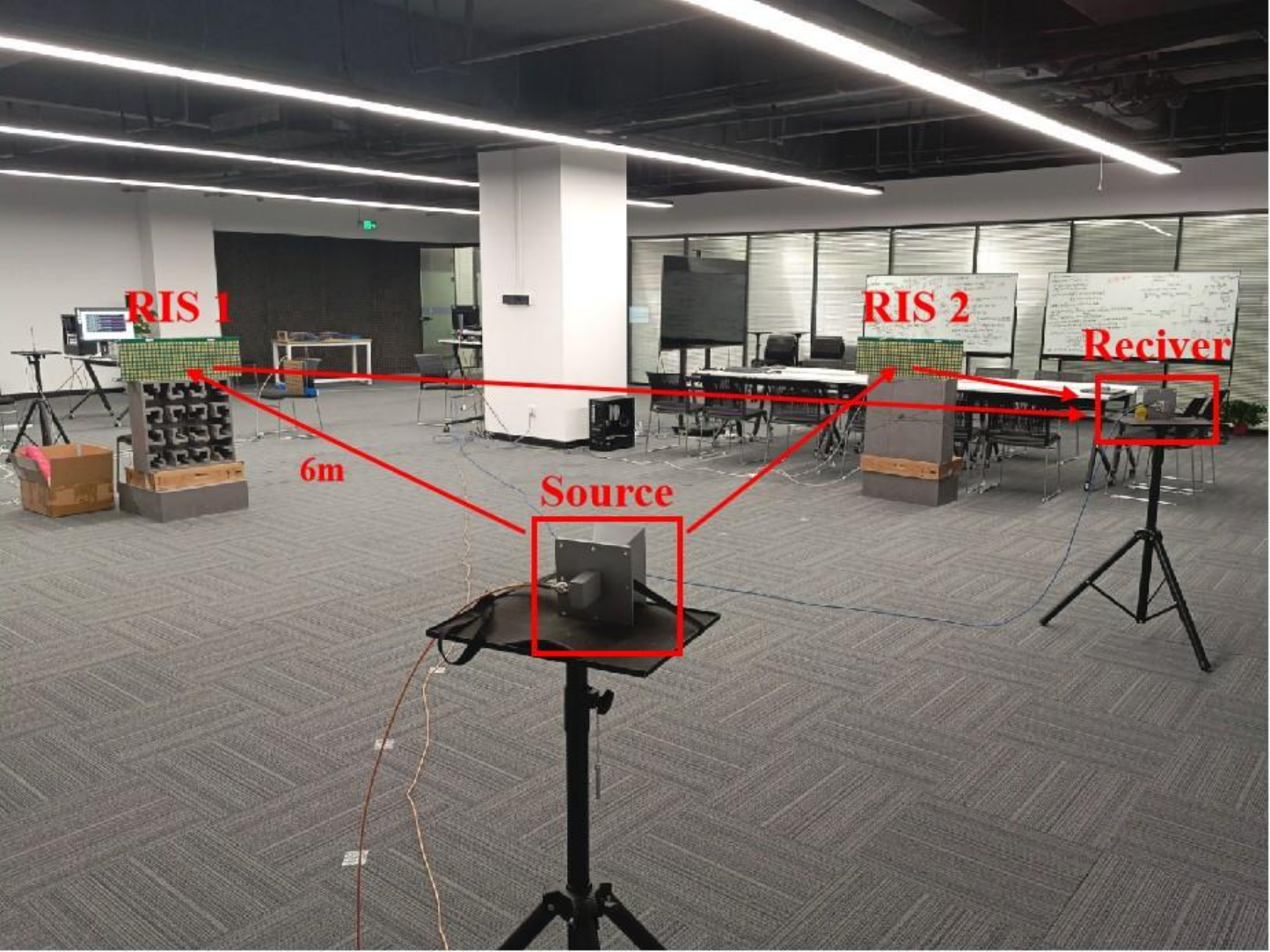}}
\caption{Experiment scenes.}
\label{Experiment scene}
\end{figure}

\begin{figure}
\centering
\subfloat[Images reconstructed in the chamber.]{
\includegraphics[width=.42\columnwidth]{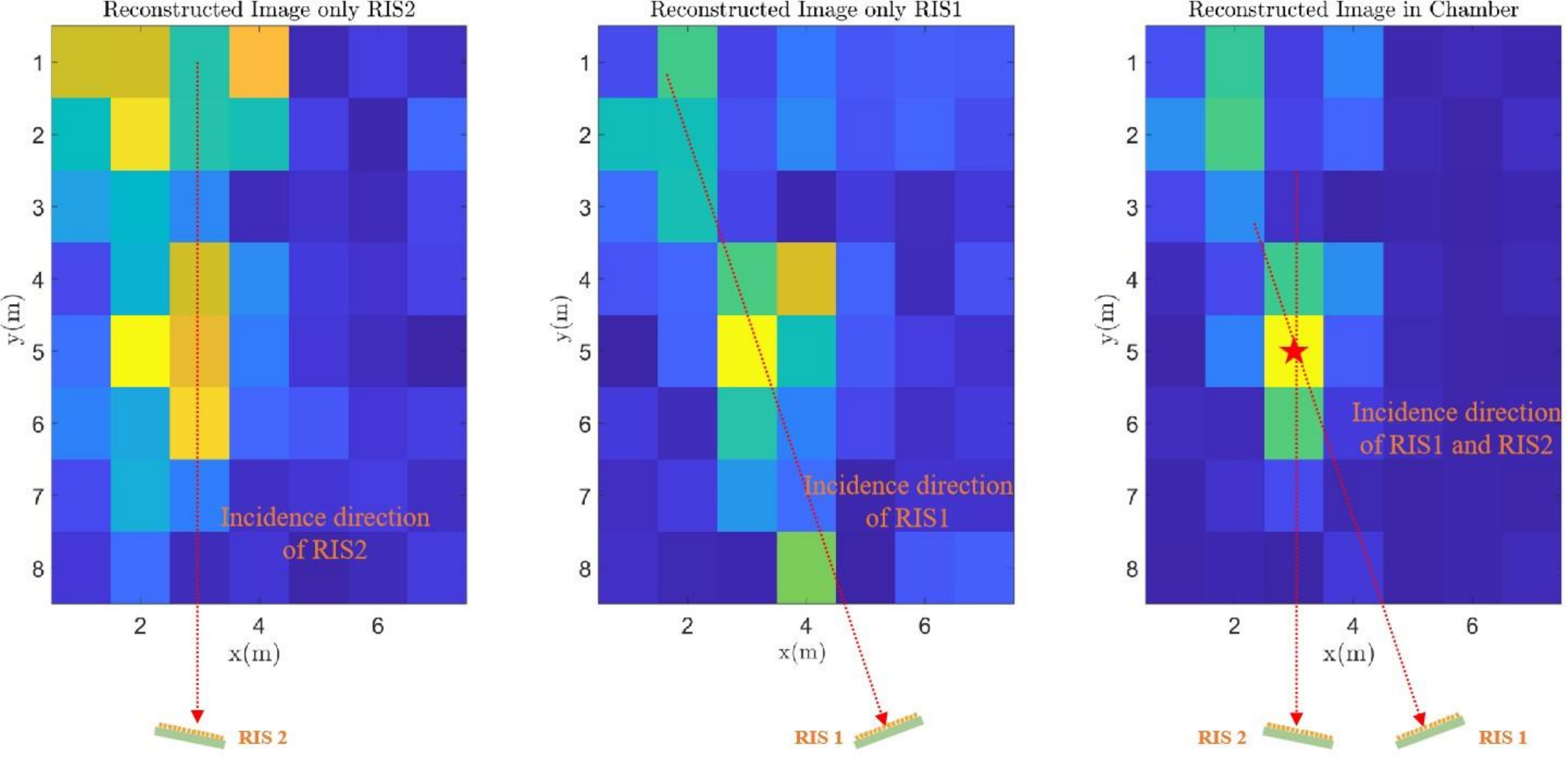}} 
\subfloat[Images reconstructed in the office.]{
\includegraphics[width=.47\columnwidth]{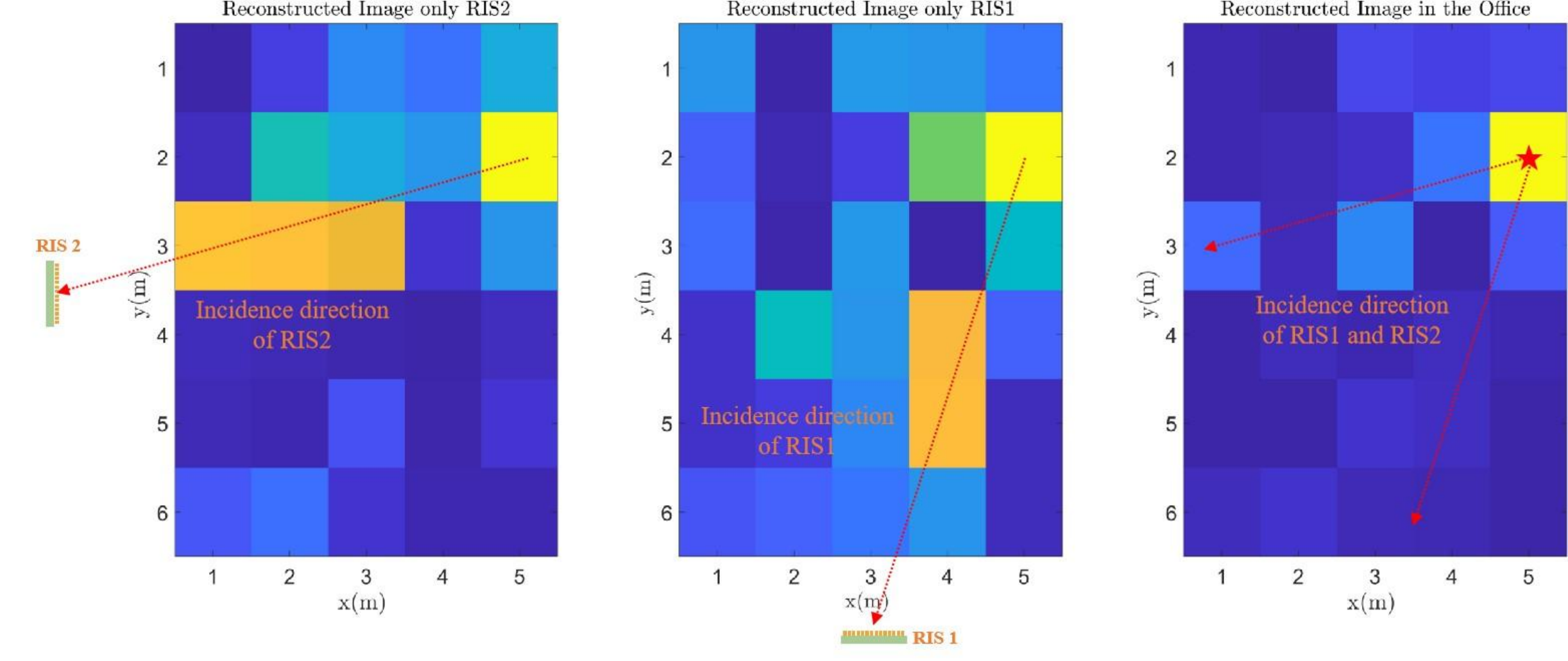}}
\caption{Reconstructed results.}
\label{Reconstructed results}
\end{figure}

In the first scene, the SoI covers a range of 7m$\times$8m spaced at 1m intervals in the microwave chamber. T=500 random samples from RIS1 are enabled to complete the data acquisition of the SoI, and similarly, RIS2 is also activated. The imaging results are shown in Fig\ref{Reconstructed results}(a), where the left and middle sub-images are recovered by RIS2 and RIS1, respectively. The right subplot is recovered using both RISs together. The estimated true source of the signal is marked with the red pentagram. In the second scene, we conduct an experimental validation closer to a real scenario in our office. The SoI covers a range of 5m$\times$6m spaced at 1m intervals. The imaging results are shown in Fig\ref{Reconstructed results}(b). The amplitude-only imaging algorithm can reconstruct an image when the phase of complex-valued signals is unknown. Compared to imaging with a single RIS, the use of multiple RISs can image SoI more clearly, and the imaging accuracy is greatly improved.

\section{Conclusion} 
We explore application scenarios for RIS-aided wireless regional imaging. In this paper, we introduce the multiple RISs linear channel modeling framework. 3D modeling from far-field scenarios is considered. We theoretically analyze the constraints on the rank of the imaging matrix. Simulation modeling shows that a sufficiently large RIS, a relatively reasonable position concerning the SoI, and a distributed RIS system can provide positive gains for imaging. To solve the problem that the phase cannot be acquired, we propose an amplitude-only imaging algorithm. The performance verification of the imaging algorithm is carried out by experiment. It is believed that RIS-aided regional imaging will inject more vitality into the 6G and meta-universe era. 
\appendices
\section{Proof of Theorem~\ref{T:1}}\label{S:Appendix_A}
\begin{proof} \label{proof1}
We denote $\text{rank}\left( {\mathbf{H_{K,T}}} \right)$ as follows:
\small 
\begin{equation} \label{proofeq1}
    \begin{aligned}
    \text{rank}\left( {\mathbf{H_{K,T}}} \right)
    &= \text{rank}\left(\begin{bmatrix}
                {\mathbf{H}}_{1,\mathbf{T}}\\
                \vdots \\
                {\mathbf{H}}_{k,\mathbf{T}}\\
             \end{bmatrix}\right) 
    \le \text{rank}\left( \sum\limits_{k = 1}^K
      {\begin{bmatrix} 
         h^s_k \mathbf{a}^s_k \mathbf{\Omega}_{k,t=1} \mathbf{A}^i_k \mathbf{H}^i_k    \\
                   \vdots                         \\
         h^s_k \mathbf{a}^s_k \mathbf{\Omega}_{k,t=T} \mathbf{A}^i_k \mathbf{H}^i_k    \\
      \end{bmatrix}}  
    \right) \\
    \end{aligned}
\end{equation}
where $\mathbf{a}^s_k = \begin{bmatrix}
        { e^{ j 2 \pi \mathbf{p}_1^{\top} \mathbf{u}  / \lambda } \cdots e^{ j 2 \pi \mathbf{p}_N^{\top} \mathbf{u}  / \lambda }}
        \end{bmatrix}
        $ and $\mathbf{\Omega}_{k,t} = \text{diag}(\gamma_{k,1} e^{j \Omega_{k,1}},\dots,\gamma_{k,N} e^{j \Omega_{k,N}} ).$
For simplicity of expression, we set $ e^{ j 2 \pi \mathbf{p}_n^{\top} \mathbf{u}  / \lambda } = a_n$ and ${\gamma_{k,n} e^{j \Omega_{k,n}}}=\omega_{k,n}$. So, the formula \ref{proofeq1} can be denoted as: 
\begin{equation} \label{proofeq2}
    \begin{aligned}
    & \text{rank}\left( \sum\limits_{k = 1}^K
        \begin{bmatrix} 
        h^s_k
        \begin{bmatrix}
        { a_1  \cdots a_N }
        \end{bmatrix}
         {\begin{bmatrix}
          \omega_{k,1}&{}&{}\\
          {}& \ddots &{}\\
          {}&{}&\omega_{k,N}
        \end{bmatrix}}_{t=1}
        \mathbf{A}^i_k \mathbf{H}^i_k \\
        {\vdots} \\
        h^s_k
        \begin{bmatrix}
         a_1  \cdots a_N
        \end{bmatrix}  
         \begin{bmatrix}
          \omega_{k,1}&{}&{}\\
          {}& \ddots &{}\\
          {}&{}&\omega_{k,N}
        \end{bmatrix}_{t=T}
        \mathbf{A}^i_k \mathbf{H}^i_k \\
      \end{bmatrix} \right) \\
    &= \text{rank}\left( \sum\limits_{k = 1}^K
      \begin{bmatrix}
         h^s_k
        \begin{bmatrix}
        \omega_{k,1} \cdots \omega_{k,N} 
        \end{bmatrix}_{t = 1}
        \begin{bmatrix}
          a_1  &{}&{} \\
          {}& \ddots &{}\\
          {}&{}& a_N
        \end{bmatrix}
        \mathbf{A}^i_k \mathbf{H}^i_k  \\
       \vdots \\
        h^s_k
        \begin{bmatrix}
        \omega_{k,1} \cdots \omega_{k,N} 
        \end{bmatrix}_{t = T}
        \begin{bmatrix}
          a_1  &{}&{} \\
          {}& \ddots &{}\\
          {}&{}& a_N
          
        \end{bmatrix}
        \mathbf{A}^i_k \mathbf{H}^i_k\\
       \end{bmatrix}
        \right) \\ 
        \end{aligned}
\end{equation}

\noindent
Here, we now set 
\begin{equation}
    \begin{aligned}
      h^s_k
        \begin{bmatrix}
          e^{ j 2 \pi \mathbf{p}_1^{\top} \mathbf{u}  / \lambda } &  & \\
           & \ddots & \\
           &  & e^{ j 2 \pi \mathbf{p}_N^{\top} \mathbf{u}  / \lambda }
        \end{bmatrix}
        \mathbf{A}^i_k \mathbf{H}^i_k
       = {{\bf{P}}_{{\bf{k}}}}
    = \begin{bmatrix}
      {{\bf{p}}_1} \\
       \vdots \\
      {{\bf{p}}_N}
      \end{bmatrix} .\\
       \end{aligned}
\end{equation}

\noindent
Then, the formula \ref{proofeq2} can be rewritten as follows, 
\begin{equation}
    \begin{aligned}
     & \text{rank}\left( 
       \sum\limits_{k = 1}^K 
       \begin{bmatrix}
      \begin{bmatrix}
      \omega_{k,1}& \cdots &\omega_{k,N}
      \end{bmatrix}_{t = 1}{\bf{P}_k}\\
       \vdots \\
      \begin{bmatrix}
      \omega_{k,1}& \cdots &\omega_{k,N}
      \end{bmatrix}_{t = T}{\bf{P}_k}
      \end{bmatrix}  \right)\\
      &= \text{rank}\left( {\sum\limits_{k = 1}^K \begin{bmatrix}
      {{{\begin{bmatrix}
      {\omega_{k,1}^{t = 1}}& \cdots &{\omega_{k,N}^{t = 1}}\\
       \vdots & \ddots & \vdots \\
      {\omega_{k,1}^{t = T}}& \cdots &{\omega_{k,N}^{t = T}}
      \end{bmatrix}}_{T \cdot N}}
      \begin{bmatrix}
      {{\bf{p}}_1} \\
       \vdots \\
      {{\bf{p}}_N}
      \end{bmatrix}^k}_{N \cdot M} \end{bmatrix} } \right) \\
      &= \text{rank}\left( {\sum\limits_{k = 1}^K 
        {\bf{W}}^k
        {\bf{P}}^k
        }\right)
      \le K \text{rank} \left( 
        {\bf{W}}^k
        {\bf{P}}^k \right)
      \le K \min \left\{ {T,N,M} \right\}. 
    \end{aligned}
\end{equation}
Therefore, the rank of imaging matrix $\mathbf{H_{K,T}}$ with K receivers is bouned by the total number of samples K$\times$T, the total number of RIS units K$\times$N, and the number of unknown grid points M.
\end{proof}
\bibliographystyle{IEEEtran}
\bibliography{Reference}

\begin{thebibliography}{10}
\providecommand{\url}[1]{#1}
\csname url@samestyle\endcsname
\providecommand{\newblock}{\relax}
\providecommand{\bibinfo}[2]{#2}
\providecommand{\BIBentrySTDinterwordspacing}{\spaceskip=0pt\relax}
\providecommand{\BIBentryALTinterwordstretchfactor}{4}
\providecommand{\BIBentryALTinterwordspacing}{\spaceskip=\fontdimen2\font plus
\BIBentryALTinterwordstretchfactor\fontdimen3\font minus \fontdimen4\font\relax}
\providecommand{\BIBforeignlanguage}[2]{{%
\expandafter\ifx\csname l@#1\endcsname\relax
\typeout{** WARNING: IEEEtran.bst: No hyphenation pattern has been}%
\typeout{** loaded for the language `#1'. Using the pattern for}%
\typeout{** the default language instead.}%
\else
\language=\csname l@#1\endcsname
\fi
#2}}
\providecommand{\BIBdecl}{\relax}
\BIBdecl

\bibitem{ITUR2022}
ITUR, ``Future technology trends of terrestrial international mobile telecommunications systems towards 2030 and beyond,'' 2022.

\bibitem{wen2019survey}
F.~Wen, H.~Wymeersch, B.~Peng, W.~P. Tay, H.~C. So, and D.~Yang, ``A survey on 5g massive mimo localization,'' \emph{Digital Signal Processing}, vol.~94, pp. 21--28, 2019.

\bibitem{basar2019wireless}
E.~Basar, M.~Di~Renzo, J.~De~Rosny, M.~Debbah, M.-S. Alouini, and R.~Zhang, ``Wireless communications through reconfigurable intelligent surfaces,'' \emph{IEEE access}, vol.~7, pp. 116\,753--116\,773, 2019.

\bibitem{song2023intelligent}
X.~Song, J.~Xu, F.~Liu, T.~X. Han, and Y.~C. Eldar, ``Intelligent reflecting surface enabled sensing: Cram{\'e}r-rao bound optimization,'' \emph{IEEE Transactions on Signal Processing}, 2023.

\bibitem{chu2021intelligent}
Z.~Chu, Z.~Zhu, F.~Zhou, M.~Zhang, and N.~Al-Dhahir, ``Intelligent reflecting surface assisted wireless powered sensor networks for internet of things,'' \emph{IEEE Transactions on Communications}, vol.~69, no.~7, pp. 4877--4889, 2021.

\bibitem{liu2022tdoa}
R.~Liu, M.~Jian, and W.~Zhang, ``A tdoa based positioning method for wireless networks assisted by passive ris,'' in \emph{2022 IEEE Globecom Workshops (GC Wkshps)}.\hskip 1em plus 0.5em minus 0.4em\relax IEEE, 2022, pp. 1531--1536.

\bibitem{meng2022intelligent}
K.~Meng, Q.~Wu, R.~Schober, and W.~Chen, ``Intelligent reflecting surface enabled multi-target sensing,'' \emph{IEEE Transactions on Communications}, vol.~70, no.~12, pp. 8313--8330, 2022.

\bibitem{rinchi2022single}
O.~Rinchi, A.~Elzanaty, and A.~Alsharoa, ``Single-snapshot localization for near-field ris model using atomic norm minimization,'' in \emph{GLOBECOM 2022-2022 IEEE Global Communications Conference}.\hskip 1em plus 0.5em minus 0.4em\relax IEEE, 2022, pp. 2432--2437.

\bibitem{zhang2021metalocalization}
H.~Zhang, H.~Zhang, B.~Di, K.~Bian, Z.~Han, and L.~Song, ``Metalocalization: Reconfigurable intelligent surface aided multi-user wireless indoor localization,'' \emph{IEEE Transactions on Wireless Communications}, vol.~20, no.~12, pp. 7743--7757, 2021.

\bibitem{he2022high}
Y.~He, D.~Zhang, and Y.~Chen, ``High-resolution wifi imaging with reconfigurable intelligent surfaces,'' \emph{IEEE Internet of Things Journal}, vol.~10, no.~2, pp. 1775--1786, 2022.

\bibitem{cui2022near}
M.~Cui, Z.~Wu, Y.~Lu, X.~Wei, and L.~Dai, ``Near-field mimo communications for 6g: Fundamentals, challenges, potentials, and future directions,'' \emph{IEEE Communications Magazine}, vol.~61, no.~1, pp. 40--46, 2022.

\bibitem{mi2023towards}
T.~Mi, J.~Zhang, R.~Xiong, Z.~Wang, P.~Zhang, and R.~C. Qiu, ``Towards analytical electromagnetic models for reconfigurable intelligent surfaces,'' \emph{IEEE Transactions on Wireless Communications}, 2023.

\bibitem{ferber2021resilience}
A.~Ferber, K.~Luh, and G.~McKinley, ``Resilience of the rank of random matrices,'' \emph{Combinatorics, Probability and Computing}, vol.~30, no.~2, pp. 163--174, 2021.

\bibitem{9721164}
L.~Song, L.~Wang, M.~H.~Kim, and H.~Huang, ``High-accuracy image formation model for coded aperture snapshot spectral imaging,'' \emph{IEEE Transactions on Computational Imaging}, vol.~8, pp. 188--200, 2022.

\bibitem{zha2020doa}
S.~Zha, M.~Lin, H.~Liu, Z.~Wu, J.~Liu, and B.~Deng, ``Doa estimation with programmable metasurface,'' in \emph{2020 9th Asia-Pacific Conference on Antennas and Propagation (APCAP)}.\hskip 1em plus 0.5em minus 0.4em\relax IEEE, 2020, pp. 1--2.

\bibitem{xiong2023ris}
R.~Xiong, J.~Zhang, X.~Dong, Z.~Wang, J.~Liu, T.~Mi, and R.~C. Qiu, ``Ris-aided wireless communication in real-world: Antennas design, prototyping, beam reshape and field trials,'' \emph{arXiv preprint arXiv:2303.03287}, 2023.

\end{thebibliography}

\end{document}